\documentclass[11pt,letterpaper]{article}
\usepackage{standalone}
\usepackage{etoolbox}

\newbool{sodasubm}
\setbool{sodasubm}{false}

\usepackage[letterpaper,margin=1in]{geometry}
\usepackage{amsthm}
\usepackage{amsmath}
\usepackage{amssymb}
\usepackage{thmtools}
\usepackage{thm-restate}
\usepackage{times}
\usepackage{mathtools}

\usepackage{placeins}

\usepackage[table]{xcolor}

\definecolor{linkcol}{rgb}{0,0,0.38}
\definecolor{citecol}{rgb}{0.8,0,0}
\definecolor{urlcol}{rgb}{0.1,0.35,0}

\usepackage[pdfpagelabels,bookmarks=false,hyperfootnotes=false]{hyperref}
\hypersetup{colorlinks, linkcolor=linkcol, citecolor=citecol, urlcolor=urlcol}

\DeclareFontFamily{U}{BOONDOX-calo}{\skewchar\font=45 }
\DeclareFontShape{U}{BOONDOX-calo}{m}{n}{
  <-> s*[1.05] BOONDOX-r-calo}{}
\DeclareFontShape{U}{BOONDOX-calo}{b}{n}{
  <-> s*[1.05] BOONDOX-b-calo}{}
\DeclareMathAlphabet{\link}{U}{BOONDOX-calo}{m}{n}
\DeclareMathAlphabet{\blink}{U}{BOONDOX-calo}{b}{n}

\usepackage[utf8]{inputenc}

\usepackage[
backend=biber,
style=alphabetic,
citestyle=alphabetic,
maxalphanames=4,
maxcitenames=99,
mincitenames=98,
maxbibnames=99,
giveninits=true,
]{biblatex}

\usepackage[nameinlink]{cleveref}

\usepackage{lmodern}

\usepackage{bbm}
\usepackage{thm-restate}

\newtheoremstyle{light} %
    {\topsep}                    %
    {\topsep}                    %
    {\itshape}                   %
    {}                           %
    {\scshape}                   %
    {.}                          %
    {.5em}                       %
    {}  %

\newtheorem{theorem}{Theorem}[section]
\newtheorem{lemma}[theorem]{Lemma}

\newtheorem{proposition}[theorem]{Proposition}
\newtheorem{definition}[theorem]{Definition}

\newtheorem{corollary}[theorem]{Corollary}
\newtheorem{conjecture}[theorem]{Conjecture}

\newtheorem{observation}[theorem]{Observation}

\theoremstyle{light}

\makeatletter
\if@cref@capitalise
\crefname{claiminproof}{Claim}{Claims}
\else
\crefname{claiminproof}{claim}{claims}
\fi

\if@cref@capitalise
\crefname{algocf}{Algorithm}{Algorithms}
\else
\crefname{algocf}{algorithm}{algorithms}
\fi

\if@cref@capitalise
\crefname{conjecture}{Conjecture}{Conjectures}
\else
\crefname{conjecture}{conjecture}{conjectures}
\fi

\if@cref@capitalise
\crefname{thm}{Theorem}{Theorems}
\else
\crefname{thm}{theorem}{theorems}
\fi

\if@cref@capitalise
\crefname{lem}{Lemma}{Lemmas}
\else
\crefname{lem}{lemma}{lemmas}
\fi

\makeatother

\usepackage{url}
\usepackage{array}

\usepackage{setspace}

\usepackage{wrapfig}

\makeatletter
\newcommand{\labeltarget}[1]{\Hy@raisedlink{\hypertarget{#1}{}}}
\makeatother

\usepackage{upgreek}

\usepackage{complexity}

\usepackage[inline]{enumitem}
\usepackage{moreenum}

\usepackage[super]{nth}

\usepackage{bm}

\usepackage{xspace}

\usepackage{ifthen}

\usepackage[immediate]{silence}
\usepackage{sectsty}
\DeactivateWarningFilters[tmp]
\allsectionsfont{\boldmath}

\setlist[enumerate]{nosep,topsep=0.1em}
\setlist[enumerate,1]{label=(\roman*), leftmargin=2.2em}
\setlist[itemize]{nosep,topsep=0.3em}

\usepackage[textsize=footnotesize, color=blue!30!white]{todonotes}
\ifbool{sodasubm}{
}{

}

\usepackage{xfrac}

\usepackage{tikz}
\usetikzlibrary{calc}
\usetikzlibrary{math}
\usetikzlibrary{shapes.geometric}
\usetikzlibrary{arrows}
\usetikzlibrary{decorations.pathreplacing, decorations.pathmorphing}

\usepackage{tcolorbox}
\tcbuselibrary{skins}

\usepackage{float}
\usepackage{graphicx}
\graphicspath{{graphics/}}
\makeatletter
\newcommand\appendtographicspath[1]{%
  \g@addto@macro\Ginput@path{#1}%
}
\makeatother

\usepackage[margin=10pt,font=small,labelfont=bf,skip=-5pt]{caption}
\usepackage{subcaption}

\usepackage[vlined,ruled,algo2e]{algorithm2e}
\SetAlgoSkip{bigskip}
\SetAlgoInsideSkip{smallskip}

\usepackage{mdframed}

\let\truehypersetup\hypersetup
\renewcommand\hypersetup[1]{}
\usepackage{bigfoot}
\let\hypersetup\truehypersetup
\DeclareMathOperator{\head}{head}
\DeclareMathOperator{\tail}{tail}
\DeclareMathOperator{\succelem}{succ}

\DeclareMathOperator{\flow}{flow}

\renewcommand{\epsilon}{\varepsilon}

\newcommand{\abs}[1]{\left\lvert #1 \right\rvert}

\newcommand{\Fin}{F_{{\textrm{in}}}}
\newcommand{\Pcc}{P_{{\textrm{cc}}}}
\newcommand{\Pclo}{P_{{\textrm{clo}}}}
\newcommand{\acc}{a_{{\textrm{cc}}}}
\newcommand{\aclo}{a_{{\textrm{clo}}}}
\newcommand{\bcc}{b_{{\textrm{cc}}}}
\newcommand{\bclo}{b_{{\textrm{clo}}}}
\renewcommand{\P}{\text{P}}
\renewcommand{\NP}{\text{NP}}
\newcommand{\inp}{\langle \mathrm{input} \rangle}

\makeatletter
\let\@@pmod\pmod
\DeclareRobustCommand{\pmod}{\@ifstar\@pmods\@@pmod}
\def\@pmods#1{\mkern8mu({\operator@font mod}\mkern 6mu#1)}
\makeatother

\makeatletter
\let\@@mod\mod
\DeclareRobustCommand{\mod}{\@ifstar\@mods\@@mod}
\def\@mods#1{\mkern8mu{\operator@font mod}\mkern 6mu#1}
\makeatother

\definecolor{green}{rgb}{0.4,0.85,0.6}
 
\makeatletter
\def\@fnsymbol#1{\ensuremath{\ifcase#1\or *\or %
\ddagger\or
    \mathsection\or \mathparagraph\or \|\or **\or \dagger\dagger
    \or \ddagger\ddagger \else\@ctrerr\fi}}
\makeatother

\title{Single-Source Unsplittable Flows in Planar Graphs%
\ifbool{sodasubm}{}{%
\thanks{This project received funding from Swiss National Science Foundation grant 200021\_184622 and the European Research Council (ERC) under the European Union's Horizon 2020 research and innovation programme (grant agreement No 817750).}}
}

\ifbool{sodasubm}{
\author{}
}{
\author{
Vera Traub\thanks{
Research Institute for Discrete Mathematics and Hausdorff Center for Mathematics, University of Bonn.
Email: \href{mailto:traub@dm.uni-bonn.de}%
{traub@dm.uni-bonn.de}.
}
\and
Laura Vargas Koch\thanks{
Department of Mathematics, ETH Zurich, Zurich, Switzerland.
Email: \href{mailto:laura.vargas@ifor.math.ethz.ch}%
{laura.vargas@ifor.math.ethz.ch}.}
\and
Rico Zenklusen\thanks{
Department of Mathematics, ETH Zurich, Zurich, Switzerland.
Email: \href{mailto:ricoz@ethz.ch}%
{ricoz@ethz.ch}.}
}
}

\date{}

\begin{document}

\maketitle
\thispagestyle{empty}
\addtocounter{page}{-1}
\enlargethispage{-1cm}

\begin{abstract}
The single-source unsplittable flow (SSUF) problem asks to send flow from a common source to different terminals with unrelated demands, each terminal being served through a single path.
One of the most heavily studied SSUF objectives is to minimize the violation of some given arc capacities.
A seminal result of Dinitz, Garg, and Goemans showed that, whenever a fractional flow exists respecting the capacities, then there is an unsplittable one violating the capacities by at most the maximum demand.
Goemans conjectured a very natural cost version of the same result, where the unsplittable flow is required to be no more expensive than the fractional one.
This intriguing conjecture remains open.
More so, there are arguably no non-trivial graph classes for which it is known to hold.

We show that a slight weakening of it (with at most twice as large violations) holds for planar graphs.
Our result is based on a connection to a highly structured discrepancy problem, whose repeated resolution allows us to successively reduce the number of paths used for each terminal, until we obtain an unsplittable flow.
Moreover, our techniques also extend to simultaneous upper and lower bounds on the flow values.
This also affirmatively answers a conjecture of Morell and Skutella for planar SSUF.
\end{abstract}
 
\ifbool{sodasubm}{}{
\begin{tikzpicture}[overlay, remember picture, shift = {(current page.south east)}]
\begin{scope}[shift={(-1.1,2.5)}]
\def\hd{2.5}
\node at (-2.15*\hd,0) {\includegraphics[height=0.7cm]{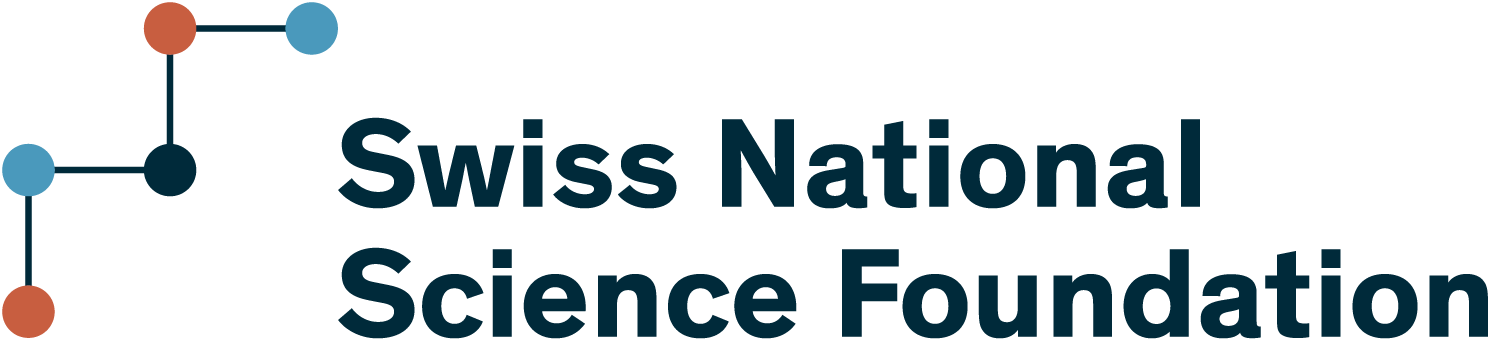}};
\node at (-\hd,0) {\includegraphics[height=1.0cm]{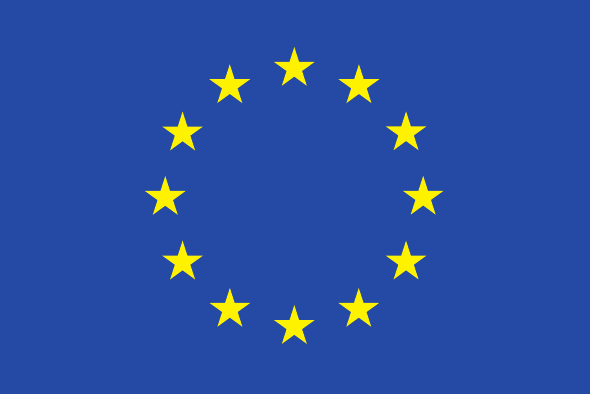}};
\node at (-0.2*\hd,0) {\includegraphics[height=1.2cm]{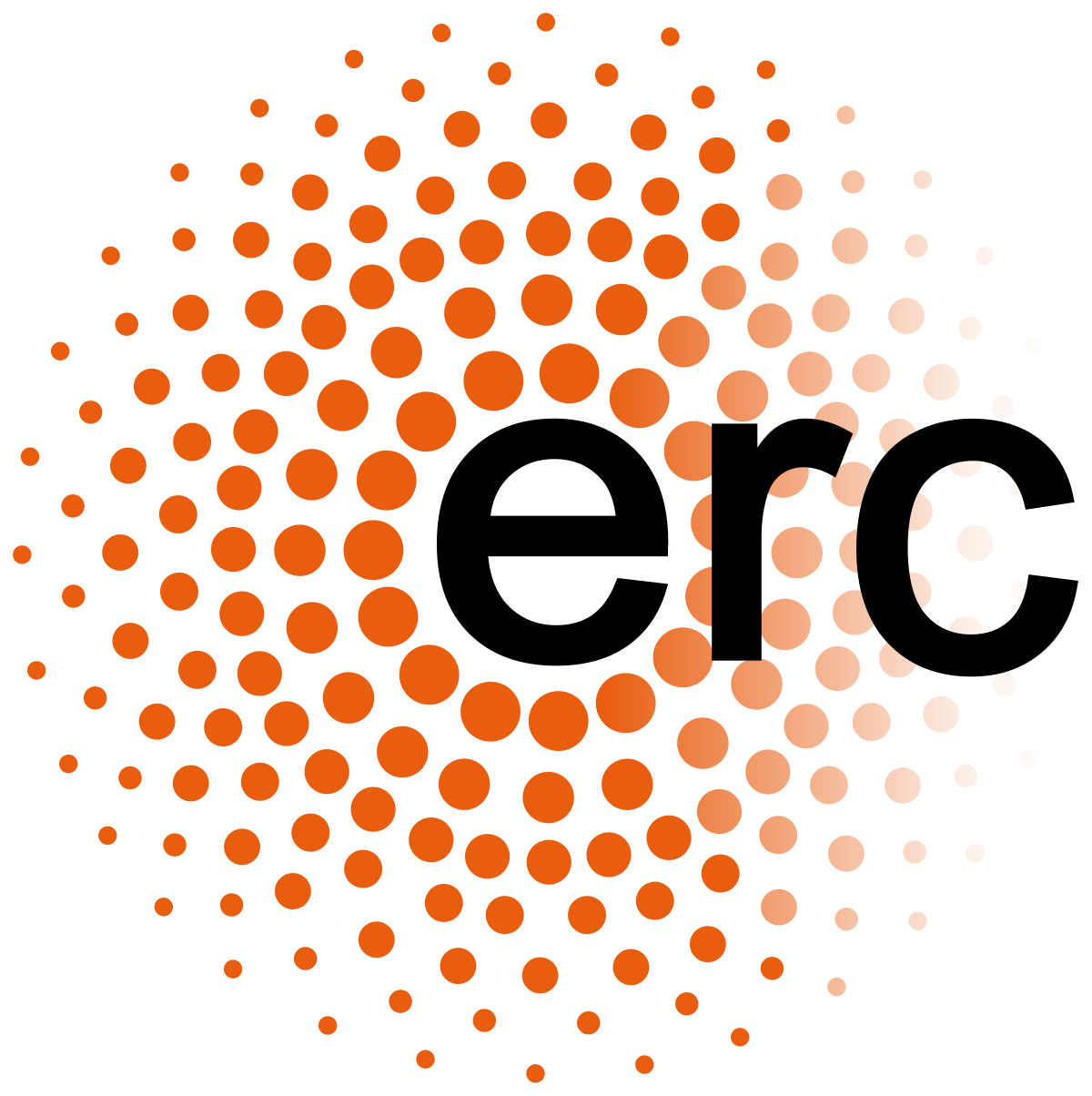}};
\end{scope}
\end{tikzpicture}
}

\clearpage

\section{Introduction}\label{sec:intro}

Flow problems are among the most classical combinatorial optimization problems.
To their prominent role in the field significantly contributed both their wide, and sometimes surprising, range of applications and also the existence of very fast algorithms (in practice as well as in theory) for basic flow optimization problems.
(We refer the interested reader to the textbooks of \textcite{ahuja_1993_network}, and \textcite{williamson_2019_network} for more information.)
In their most canonical version, flow from a source to a sink can be split over many paths. 
Being able to split flows naturally allows for casting flow problems---even with multiple sources, sinks, and commodities---as linear programs.
Flow problems become much harder, and are also much less understood, as soon as the flow is not allowed to be split, i.e., is unsplittable.

In this work, we consider a heavily studied variant thereof, namely the \emph{single-source unsplittable flow problem} (SSUF), originally introduced by \textcite{kleinberg_1996_approximation,kleinberg_1996_singlesource}.
Here, we are given a directed graph $G=(V,A)$, with arc capacities $u\in \mathbb{Q}_{\geq 0}^A$, a single source $s\in V$, and terminals/sinks $T\subseteq V$, where each $t\in T$ has a demand $d_t \in \mathbb{Q}_{\geq 0}$.
(The notions of sink and terminal are used interchangeably in the literature.)
For simplicity, and without loss of generality, we assume that the source $s$ and sinks $T$ are all distinct.
Ideally, one would like to route the demanded flow from the source to each terminal on a single path such that all capacities are respected.
Formally, this corresponds to determining one $s$-$t$ path $P^t\subseteq A$ for each $t\in T$ such that $\sum_{\substack{t\in T: a\in P^t}}d_t \leq u(a)$ for each $a\in A$.
Deciding whether an \emph{unsplittable flow}, i.e., such a family $\mathcal{P}=\{P^t\}_{t\in T}$ of paths, exists, can easily be seen to be NP-hard.
Actually, even in a $2$-vertex graph, one can reduce Bin Packing or Subset Sum to it.
Thus, significant attention has been devoted to the setting where one is allowed to exceed the capacities and the goal is to minimize the largest capacity violation.
This is also the setting we consider here.

The research in this area has been largely influenced by a seminal result of \textcite{dinitz_1999_singlesource} and a subsequent conjecture of Goemans, which we discuss next.
These are statements about the existence of unsplittable flows with limited capacity violation assuming that there exists a \emph{splittable flow} $x\in \mathbb{Q}^A$ 
that respects the capacities, i.e., such a vector $x$ satisfies $x(a) \leq u(a)$ for all arcs $a \in A$ and
\begin{equation*}
x(\delta^+(v)) - x(\delta^-(v)) = \begin{cases}
\sum_{t\in T}d_t & \text{if } v=s,\\
-d_v & \text{if } v\in T,\\
0 & \text{if } v\in V\setminus (\{s\}\cup T),
\end{cases}
\end{equation*}
where $\delta^-(v)$ and $\delta^+(v)$ denote the sets of all arcs entering and leaving $v$, respectively.
More precisely, these statements assume the existence of a splittable flow $x\in \mathbb{Q}^A_{\geq 0}$, and the capacity violation of an unsplittable flow is measured by the difference between the flow value of the unsplittable flow and the value of the splittable one.
In short, this can be interpreted as a worst-case assumption where the capacity $u(a)$ of an arc $a$ is equal to $x(a)$.
We thus define an SSUF instance as follows.
\begin{definition}[Single-source unsplittable flow (SSUF) instance]
A \emph{single-source unsplittable flow (SSUF)} instance is a tuple $(G=(V,A),s,T,d,x)$, where $G$ is a directed graph with source $s\in V$, terminals $T\subseteq V$ with corresponding demands $d\in \mathbb{Q}_{\geq 0}^T$, and a splittable flow $x\in \mathbb{Q}_{\geq 0}^A$.
The source $s$ and terminals $T$ are assumed to be distinct.
\end{definition}
The result of \textcite{dinitz_1999_singlesource}, stated below, was the first to establish an additive capacity violation of $O(d_{\max})$, where $d_{\max}\coloneqq \max\{d_t\colon t\in T\}$ is the largest demand, whenever a splittable flow exists.
For ease of notation, given an unsplittable flow $\mathcal{P}=\{P^t\}_{t\in T}$ and an arc $a\in A$, we denote by
\begin{equation*}
\flow_{\mathcal{P}}(a)\coloneqq \sum_{t\in T\colon a\in P^t} d_t
\end{equation*}
the total flow value of the unsplittable flow $\mathcal{P}$ that traverses $a$.
\begin{theorem}[\cite{dinitz_1999_singlesource}]\label{thm:dinitzGargGoemans}
Given an SSUF instance $(G,s,T,u,d,x)$, one can compute in polynomial time an unsplittable flow $\mathcal{P}=\{P^t\}_{t\in T}$ with $\flow_{\mathcal{P}}(a) \leq x(a) + d_{\max}$ for all $a\in A$.
\end{theorem}
Shortly thereafter, Goemans conjectured that the following stronger, cost-enhanced version of the same result holds.
\begin{conjecture}[Goemans]\label{conj:goemans}
Given an SSUF instance $(G,s,T,u,d,x)$ and a cost vector $c\in \mathbb{Q}^A_{\geq 0}$, one can compute in polynomial time an unsplittable flow $\mathcal{P}=\{P^t\}_{t\in T}$ with
\begin{enumerate}
\item $\flow_{\mathcal{P}}(a) \leq x(a) + d_{\max}$ for all $a\in A$, and \label{it: flow_upper_bounds}
\item cost at most $c^Tx$, i.e., $\sum_{a\in A} c(a) \flow_{\mathcal{P}}(a) \leq \sum_{a\in A} c(a) x(a)$.
\end{enumerate}
\end{conjecture}
Note that one can assume $G$ to be acyclic in \Cref{conj:goemans}.
Indeed, if $G$ contains directed cycles, then flow along such cycles can be reduced and zero-flow arcs can be removed.
This leads to an acyclic instance that is no easier than the original one.

The following even stronger version of Goemans' conjecture, which includes both upper and lower bounds on the flow values, has been stated by \textcite{morell_2022_single}.
\begin{conjecture}[\cite{morell_2022_single}]\label{conj:morellAndSkutella}
Given an SSUF instance $(G,s,T,u,d,x)$ on an acyclic graph $G$ and a cost vector $c\in \mathbb{Q}^A_{\geq 0}$, one can compute in polynomial time an unsplittable flow $\mathcal{P}=\{P^t\}_{t\in T}$ with
\begin{enumerate}
\item $x(a) - d_{\max} \leq \flow_{\mathcal{P}}(a) \leq x(a) + d_{\max}$ for all $a\in A$, and \label{it: flow_upper_and_lower_bounds}
\item $\sum_{a\in A} c(a) \flow_{\mathcal{P}}(a) \leq \sum_{a\in A} c(a) x(a)$.
\end{enumerate}
\end{conjecture}
The fact that $G$ is explicitly assumed to be acyclic in \Cref{conj:morellAndSkutella} is due to the lower bounds.
Without acyclicity, the nature of the problem changes.
In particular, it would suddenly become important whether an unsplittable flow consists of paths or walks, because one may try to fulfill lower bounds by going along cycles.
Also, one can show that statements as \Cref{conj:morellAndSkutella} (even with a violation of $O(d_{\max})$ instead of just $d_{\max}$) cannot be obtained in general (non-acyclic) graphs (see \Cref{sec:hardness_cyclic_graphs}).

\textcite{morell_2022_single} also explicitly conjectured the following weaker version of \Cref{conj:morellAndSkutella} without costs, which also remains open and is a natural intermediate step toward \Cref{conj:morellAndSkutella}.
\begin{conjecture}[\cite{morell_2022_single}]\label{conj:morellAndSkutella_weaker}
Given an SSUF instance $(G,s,T,u,d,x)$ on an acyclic graph $G$ there exists an unsplittable flow $\mathcal{P}=\{P^t\}_{t\in T}$ with
\[ x(a) - d_{\max} \leq \flow_{\mathcal{P}}(a) \leq x(a) + d_{\max} \text{ for all } a\in A.\]
\end{conjecture}
\Cref{conj:morellAndSkutella_weaker} can be seen as the counterpart of \Cref{thm:dinitzGargGoemans} with additional lower bounds.
One motivation for \Cref{conj:morellAndSkutella_weaker} is that a positive resolution of it would have implications to both some open discrepancy and scheduling problems, on which we briefly expand in \Cref{sec:linkDiscrAndScheduling}.

All three conjectures, \Cref{conj:goemans,conj:morellAndSkutella,conj:morellAndSkutella_weaker}, remain open, and even a resolution of a weaker version of Goemans' conjecture with additive capacity violations in the order $O(d_{\max})$ (instead of the conjectured $d_{\max}$, which is the best one can hope for in terms of constants) would likely be considered a breakthrough.

Procedures leading to an unsplittable flow with a multiplicative violation, i.e., $\flow_{\mathcal{P}}(a) = O(x(a))$ for $a\in A$, have been obtained previously~\cite{%
kleinberg_1996_approximation,%
kleinberg_1996_singlesource,%
kolliopoulos_2002_approximation,%
skutella_2002_approximating,%
morell_2022_single,%
}.
Several of these approximation algorithms have also been implemented and empirically tested \cite{du_2005_implementing}.
Moreover, \textcite{skutella_2002_approximating} showed that \Cref{conj:goemans} holds when the demands are all multiples of each other.
Furthermore, techniques presented in \textcite{lenstra1990approximation} solve the problem for the special case when $G$ has a source plus two layers of vertices, with arcs only going from the source to the first layer and from the first to the second layer.

Very little is known about \Cref{conj:morellAndSkutella,conj:morellAndSkutella_weaker}.
When only dealing with lower bounds, analogous results (for acyclic graphs) can be obtained as with upper bounds only.
More precisely, \textcite{morell_2022_single} showed that for any SSUF instance $(G,s,T,u,d,x)$, there exists an unsplittable flow $\mathcal{P}=\{P^t\}_{t\in T}$ with $\flow_{\mathcal{P}}(a) \geq x(a)-d_{\max}$ for $a\in A$.
Moreover, one can also adjust the algorithm of \textcite{dinitz_1999_singlesource} to get the same lower bound guarantees.
However, the combination of both lower and upper bounds seems to be much more challenging.

Prior to this work, there was arguably no non-trivial graph class for which any of \Cref{conj:goemans,conj:morellAndSkutella,conj:morellAndSkutella_weaker} was known to hold, even if additive violations in the order of $O(d_{\max})$ are allowed instead of only $d_{\max}$.
The goal of this work is to address this gap, by showing that a slightly weaker version of \Cref{conj:morellAndSkutella} (allowing a violation of $2 d_{\max}$ instead of $d_{\max}$) and \Cref{conj:morellAndSkutella_weaker} hold for planar graphs.
This is a further positive sign regarding \Cref{conj:goemans,conj:morellAndSkutella,conj:morellAndSkutella_weaker}.
We obtain this result by connecting SSUF in planar graphs to a very well-structured discrepancy problem, whose resolution allows us to successively transform a splittable flow into another one that uses fewer paths per terminal until, eventually, an unsplittable flow is obtained.

\subsection{Our results}\label{sec:ourResults}

As mentioned, we focus on planar graphs in this work.
A planar instance of SSUF is formally defined as follows.
\begin{definition}[Planar single-source unsplittable flow (PSSUF) instance]
An SSUF instance $(G,s,T,d,x)$ is a \emph{planar single-source unsplittable flow (PSSUF)} instance if $G$ is acyclic and planar.
\end{definition}

Our first main result shows that, for planar graphs, unsplittable flows can be constructed with the lower and upper bound guarantees as claimed in \Cref{conj:morellAndSkutella_weaker}.
\begin{theorem}\label{thm:main}
Given a PSSUF instance $(G,s,T,d,x)$, there is an unsplittable flow $\mathcal{P}=\{P^t\}_{t\in T}$ with
\begin{equation*}
x(a) - d_{\max} \leq \flow_{\mathcal{P}}(a) \leq x(a) + d_{\max} \qquad \forall a\in A.
\end{equation*}
Moreover, it can be computed in time $O(\abs{V}^2)$.
\end{theorem}

We obtain an analogous result involving costs when allowing for a slightly larger lower and upper bound violation of $2 d_{\max}$ instead of $d_{\max}$.
\begin{theorem}\label{thm:main_cost}
Given a PSSUF instance $(G,s,T,d,x)$ and arc costs $c\in \mathbb{Q}_{\geq 0}^A$, there is an unsplittable flow $\mathcal{P}=\{P^t\}_{t\in T}$ satisfying
\begin{enumerate}[itemsep=0.2em]
\item $\displaystyle x(a) - 2d_{\max} \leq \flow_{\mathcal{P}}(a) \leq x(a) + 2d_{\max} \quad \forall a\in A$, and
\item $\displaystyle \sum_{a\in A} c(a) \flow_{\mathcal{P}}(a) \leq \sum_{a\in A} c(a) x(a)$.
\end{enumerate}
Moreover, it can be computed in time $O(\inp \cdot \abs{V})$, where $\inp$ denotes the input size of the instance.
\end{theorem}

\Cref{thm:main_cost} settles \Cref{conj:morellAndSkutella} for planar graphs, when allowing for a violation of $2 d_{\max}$ instead of $d_{\max}$.
Moreover, the above results are the first SSUF results to simultaneously respect upper and lower bounds with only additive errors of $O(d_{\max})$ in a non-trivial graph class.
We provide an overview of our techniques, which are based on a reduction to a well-structured discrepancy problem, in \Cref{sec:overview}.

\subsection{Connections to discrepancy problems and scheduling}\label{sec:linkDiscrAndScheduling}

SSUF is closely connected to scheduling problems.
Already \textcite{kleinberg_1996_approximation} observed that a natural special case of makespan minimization on unrelated machines can be cast as SSUF with upper bounds only. 
When allowing for upper and lower bounds, an interesting more general scheduling problem can be captured as observed by Lars Rohwedder. 
This is a non-preemptive maximum flow time minimization problem with release times, where each job can only be processed on a subset of machines. 
It remains open whether this setting allows for $O(1)$-approximations, and a constructive resolution of \Cref{conj:morellAndSkutella_weaker}, even with a violation of $O(d_{\max})$ instead of $d_{\max}$, would resolve this question.
(See \cite{morell_2022_single} for more details including a discussion of Rohwedder's reduction.)
The currently best known approximation factor is $O(\log n)$ by \textcite{bansal2015minimizing}, while the best known lower bound for the variant with unrelated machines is $3/2$ (unless $\P = \NP$) by \textcite{lenstra1990approximation}.
Moreover, \textcite{bansal2022flow} showed very recently that the integrality gap of the natural LP relaxation for this scheduling problem is $O(\sqrt{\log n})$.
(A non-constructive resolution of \Cref{conj:morellAndSkutella_weaker} would imply an $O(1)$ integrality gap for this LP relaxation.)

Furthermore, in \cite{bansal2022flow} an intriguing equivalence between the existence of an LP-based $O(1)$-approximation for this scheduling problem and the resolution of an interesting open discrepancy problem was revealed.
This discrepancy problem is a special case of the prefix version of Beck-Fiala, which is a well-studied discrepancy problem.
A positive answer to \Cref{conj:morellAndSkutella_weaker} would imply a positive answer to the open discrepancy problem, in an interesting special case, which would be of independent interest.

Interestingly, in our algorithm for PSSUF, we also rely on a prefix version of a discrepancy problem.
However, our notion of prefixes is different from the one in the prefix Beck-Fiala problem.

\subsection{Further related work}\label{sec:furtherRelatedWork}
\textcite{martens_2007_convex} took an approach toward Goemans' conjecture where they allow for slight demand modifications.
More precisely, they provide a polynomial-time procedure that, given any splittable single-source flow, writes it as a convex combination of unsplittable flows for slightly rounded demands (by a factor of at most $2$), such that the average demands used in the convex combination correspond to the original ones.
If each term in the convex combination sent the original demands, then this would imply Goemans' conjecture.
(Actually the existence of such a convex combination is equivalent to the existence of an unsplittable flow fulfilling the conditions of Goemans' conjecture.)

Also, we would like to mention that there has been extensive work on unsplittable flows in a variety of settings, including multiple sources and sinks and with different objectives.
We refer the interested reader to the survey by \textcite{kolliopoulos2007edge}, and the discussion and references in a very recent contribution of \textcite{grandoni2022ptas} to unsplittable flows on a path.
Moreover, intermediate notions between splittable and unsplittable flows have been considered.
In particular, $k$-splittable flows, which have been introduced by \textcite{baier_2005_k-splittable}, allow for splitting the flow among up to $k$ paths.
(See \cite{kolliopoulos_2005_minimum-cost,koch2008maximum,salazar_2009_single-source} for further results on $k$-splittable flows.)

\subsection{Organization of the paper}
In \Cref{sec:overview}, we provide an overview of our approach.
A key ingredient is the construction of a nice path decomposition of the fractional flow, which we discuss in \Cref{sec:nice_path_decomp}.
\Cref{sec:properties_path_decomp} discusses and proves properties of nice path decompositions, which we crucially exploit later.
In \Cref{sec:discrepancy}, we present polynomial time algorithms to solve the highly structured prefix discrepancy problem that we need to solve to determine an appropriate selection of paths for the unsplittable flow problem, both for the variant with costs and without costs.
We conclude in \Cref{sec:conclusions} and,
finally, \Cref{sec:hardness_cyclic_graphs} shows why statements like \Cref{thm:main} or \Cref{conj:morellAndSkutella} cannot be extended to cyclic graphs, even when allowing for a violation of $O(d_{\max})$ instead of $d_{\max}$.
\section{Overview}\label{sec:overview}

To simplify the presentation, we assume in the following that no terminal has an outgoing arc.
This is without loss of generality, because we can add for each $t\in T$ a new vertex $t'$ together with an arc $a_t=(t,t')$ with $x(a_t) \coloneqq d_t$ to obtain an equivalent PSSUF instance with terminal set $\{ t' : t\in T\}$.
(Note that this modification indeed preserves planarity.)

In order to prove \Cref{thm:main} and \Cref{thm:main_cost} we will proceed in two steps.
First, we carefully decompose the given flow $x$ into paths starting at $s$, i.e., we compute paths $P_1, \dots, P_{\ell}$ with weights $\lambda_i >0$ for $i\in [\ell]$ such that
\[
x=\sum_{i=1}^{\ell} \lambda_i \cdot \chi^{P_i},
\]
where $\chi^{P_i}\in \{0,1\}^A$ denotes the incidence vector of $P_i$ and $[\ell]\coloneqq\{1,\dots,\ell\}$.
(Here we identify a path with its arc set.)
Because $G$ is acyclic, we have $x(\delta^+(s))=\sum_{i=1}^{\ell} d_i$ and hence every path $P_i$ ends at some terminal.
In a second step, we then select for each terminal $t\in T$ a path $P^t$ from the set 
\[
\mathcal{P}^t \coloneqq \{ P_i : i\in[\ell] \text{ and }P_i\text{ ends at terminal }t\}.
\]
\Cref{fig:need_to_choose_paths_carefully} shows that it is crucial to choose the path decomposition of the flow $x$ carefully for this approach to have a chance to lead to an unsplittable flow as claimed in \Cref{thm:main,thm:main_cost}.

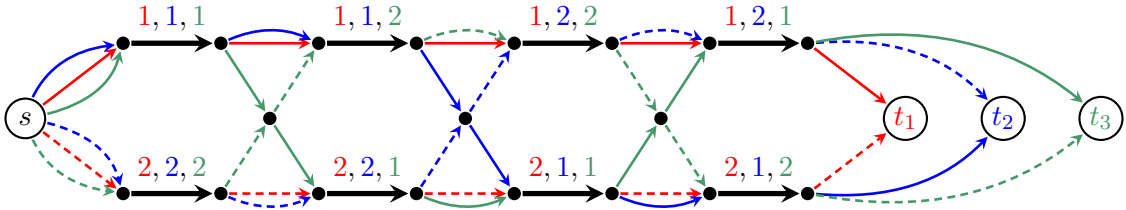
\begin{figure}[htb]
\begin{center}

\begin{tikzpicture}[xscale=1.3,
vertices/.style={circle,fill=black, draw=black, inner sep=1.5pt, thick}, outer sep=1pt,
arcs/.style={line width=1pt,-stealth},
ss/.style={fill=white,minimum size=15}
]

\colorlet{ct1}{red}
\colorlet{ct2}{blue}
\colorlet{ct3}{green!70!black}

\begin{scope}[every node/.style={vertices}]
\node[ss] (s) at (0,0) {$s$};
\node (1a) at (1,1) {};
\node (1b) at (2,1) {};
\node (2a) at (3,1) {};
\node (2b) at (4,1) {};
\node (3a) at (5,1) {};
\node (3b) at (6,1) {};
\node (4a) at (7,1) {};
\node (4b) at (8,1) {};
\node (5a) at (1,-1) {};
\node (5b) at (2,-1) {};
\node (6a) at (3,-1) {};
\node (6b) at (4,-1) {};
\node (7a) at (5,-1) {};
\node (7b) at (6,-1) {};
\node (8a) at (7,-1) {};
\node (8b) at (8,-1) {};
\node[ss,text=ct1] (t1) at (9,0) {$t_1$};
\node[ss,text=ct2] (t2) at (10,0) {$t_2$};
\node[ss,text=ct3] (t3) at (11,0) {$t_3$};
\node(p1) at (2.5,0) {};
\node(p2) at (4.5,0) {};
\node(p3) at (6.5,0) {};
\end{scope}

\begin{scope}[arcs]

\begin{scope}[line width=2]
\draw (1a) -- node[above] {$\textcolor{ct1}{1},\textcolor{ct2}{1},\textcolor{ct3}{1}$} (1b);
\draw (2a) -- node[above] {$\textcolor{ct1}{1},\textcolor{ct2}{1},\textcolor{ct3}{2}$} (2b);
\draw (3a) -- node[above] {$\textcolor{ct1}{1},\textcolor{ct2}{2},\textcolor{ct3}{2}$} (3b);
\draw (4a) -- node[above] {$\textcolor{ct1}{1},\textcolor{ct2}{2},\textcolor{ct3}{1}$} (4b);
\draw (5a) -- node[above] {$\textcolor{ct1}{2},\textcolor{ct2}{2},\textcolor{ct3}{2}$} (5b);
\draw (6a) -- node[above] {$\textcolor{ct1}{2},\textcolor{ct2}{2},\textcolor{ct3}{1}$} (6b);
\draw (7a) -- node[above] {$\textcolor{ct1}{2},\textcolor{ct2}{1},\textcolor{ct3}{1}$} (7b);
\draw (8a) -- node[above] {$\textcolor{ct1}{2},\textcolor{ct2}{1},\textcolor{ct3}{2}$} (8b);
\end{scope}

\begin{scope}[ct1]
\draw (s) -- (1a);
\draw (1b) -- (2a);
\draw (2b) -- (3a);
\draw (3b) -- (4a);
\draw (4b) -- (t1);
\end{scope}
\begin{scope}[ct1, densely dashed]
\draw (s) -- (5a);
\draw (5b) -- (6a);
\draw (6b) -- (7a);
\draw (7b) -- (8a);
\draw (8b) -- (t1);
\end{scope}
\begin{scope}[ct2]
\draw[bend left] (s) to (1a);
\draw[bend left] (1b) to (2a);
\draw (2b) to (p2);
\draw (p2) to (7a);
\draw[bend right] (7b) to (8a);
\draw[bend right] (8b) to (t2);
\end{scope}
\begin{scope}[ct2, densely dashed]
\draw[bend left] (s) to (5a);
\draw[bend right] (5b) to (6a);
\draw (6b) to (p2);
\draw (p2) to (3a);
\draw[bend left] (3b) to (4a);
\draw[bend left] (4b) to (t2);
\end{scope}
\begin{scope}[ct3]
\draw[bend right] (s) to (1a);
\draw (1b) to (p1);
\draw (p1) to (6a);
\draw[bend right] (6b) to (7a);
\draw (7b) to (p3);
\draw (p3) to (4a);
\draw[bend left] (4b) to (t3);
\end{scope}
\begin{scope}[ct3, densely dashed]
\draw[bend right] (s) to (5a);
\draw (5b) to  (p1);
\draw (p1) to (2a);
\draw[bend left] (2b) to (3a);
\draw (3b) to (p3);
\draw (p3) to (8a);
\draw[bend right] (8b) to (t3);
\end{scope}
\end{scope}

\begin{scope}

\end{scope}

\end{tikzpicture}
 \end{center}
\caption{The figure shows an example of a path decomposition of some flow $x$, which we define through the path decomposition in this example.
For each terminal $t\in T=\{{\color{red}t_1},{\color{blue}t_2},{\color{green!70!black}t_3}\}$, there are two $s$-$t$ paths $P_t^1$ and $P_t^2$ shown in the color of terminal $t$, where $P_t^1$ is drawn as solid arcs and $P_t^2$  as dashed ones.
For any combination $r\in \{1,2\}^T$ of first/second paths for each terminal, there is one black arc traversed by precisely this combination.
(The corresponding combination $r$ is highlighted above the black arcs.)
Every path has weight $1$ in our path decomposition of $x$.
Thus, $x(a)=|T|$ for each black arc $a$, and $d_t=2$ for all $t\in T$.
Then, no matter how we choose the paths $P^{t} \in \mathcal{P}^t=\{ P^1_t, P^2_t\}$,  there is an arc $a$ such that  $\flow_{\mathcal{P}}(a) = 2 |T| = x(a)+ \frac{1}{2} |T| \cdot d_{\max}$.
Note that this construction canonically extends to more than three terminals.
\label{fig:need_to_choose_paths_carefully}
}
\end{figure}

Our flow decomposition will be chosen such that the paths are pairwise non-crossing in a  geometric sense.
We will number the paths according to the order in which they leave the source $s$ (in a fixed planar embedding of our graph).
We show that our choice of the flow decomposition implies useful structural properties of the sets $\mathcal{P}^t$ (see \Cref{lem:non-interleaving}) and the sets $\{ i\in[\ell] : a\in P_i\}$ (see \Cref{lem:discr_equiv_to_interval}).\footnote{
The paths $P_i$ in $G$ will be called $\phi(P_i)$ in \Cref{lem:discr_equiv_to_interval}.
This is due to technical reasons explained in  \Cref{sec:nice_path_decomp}.}
These allow for reducing the selection of the paths $P^t\in \mathcal{P}^t$ to a highly structured discrepancy problem.

In the remainder of this section we provide a more detailed description of our approach.
We first describe the choice of our flow decomposition and the resulting structure of the sets $\mathcal{P}^t$ in \Cref{sec:overview_path_decomp} and then describe the discrepancy problem and our solution for it in \Cref{sec:overview_discrepancy}.
Finally, we show how the structure of the sets $\{ i\in[\ell] : a\in P_i\}$ allows for reducing the selection of the paths $P^t\in \mathcal{P}^t$ to this discrepancy problem.

\subsection{Choosing the path decomposition}\label{sec:overview_path_decomp}

We now provide an overview of how we choose the path decomposition of the flow $x$.
Further details and proofs will be provided in \Cref{sec:nice_path_decomp}.

We start by computing a planar embedding of the graph $G$, i.e., an embedding of the vertices and arcs in the plane such that vertices are mapped to distinct points and no two arcs intersect in a point distinct from their endpoints.
This can be done in linear time (see \cite{hopcroft_1974_efficient,chiba_1985_linear}).
Our goal will be to choose the paths $P_1, \dots, P_{\ell}$ such that they do not cross in a geometric sense.
To this end, it is helpful to first represent the paths $P_1,\dots, P_{\ell}$ in an auxiliary graph $H=(V,F)$ obtained from $G=(V,A)$, by replacing each arc $a\in A$ by a well-chosen number of parallel copies.
We call such graphs $H$ \emph{arc-split graphs} of $G$, and, for every $f\in F$, we denote by $\phi(f)\in A$ the arc in $G$ corresponding to $f$, i.e., $f$ is a parallel copy of $\phi(f)$.
The idea is to choose $H$ such that $P_1,\dots, P_{\ell}$ can be chosen to be arc-disjoint in $H$.

More precisely, given a PSSUF instance $(G,s,T,d,x)$, we start by constructing the following:
\begin{enumerate}[label=(\roman*)]
\item an arc-split graph $H=(V,F)$ of $G$ (with the associated mapping $\phi$);
\item\label{item:snNicePaths} a partition of $F$ into paths $P_1,\dots, P_{\ell}$, each of which is a path in $H$ from $s$ to some terminal $t\in T$, and such that these paths are (geometrically) non-crossing, and numbered counterclockwise around $s$;
\item\label{item:lambdas} coefficients $\lambda_1,\dots, \lambda_\ell \in \mathbb{Q}_{\geq 0}$ such that $\sum_{i=1}^{\ell}\lambda_i \chi^{\phi(P_i)} = x$.
\end{enumerate}

\noindent See \Cref{fig:arc-split_ex} for an example.

\begin{figure}[htb]
\begin{center}
\begin{tikzpicture}[xscale=1.08,
vertices/.style={circle,fill=black, draw=black, inner sep=1.5pt, thick}, outer sep=1pt,
arcs/.style={line width=1.3pt,-stealth, shorten >=1pt},
ss/.style={fill=white,minimum size=15},
p1/.style={solid},
p2/.style={dashed, dash pattern=on 8pt off 2pt},
p3/.style={densely dotted},
p4/.style={dashdotted, dash pattern=on 6pt off 2pt on 1pt off 2pt},
]

\small

\colorlet{ct1}{red!80!black}
\colorlet{ct2}{blue!80!black}
\colorlet{ct3}{green!50!black}
\colorlet{ct4}{cyan!80!black}

\begin{scope}

\begin{scope}[every node/.style={vertices},scale=0.5]

\begin{scope}[every node/.append style={ss}]
\node  (1) at (7.16,-9.04) {$s$};
\node[text=ct1]  (2) at (6.60,-15.00) {$t_1$};
\node[text=ct2]  (3) at (10.66,-11.74) {$t_2$};
\node[text=ct3]  (4) at (13.06,-6.56) {$t_3$};
\node[text=ct4]  (5) at (3.50,-6.30) {$t_4$};
\end{scope}

\node  (6) at (3.68,-11.70) {};
\node  (7) at (10.58,-7.66) {};
\node  (8) at (13.30,-9.00) {};
\node  (9) at (14.40,-15.00) {};
\node (10) at (8.00,-4.44) {};
\node (11) at (14.92,-5.38) {};
\end{scope}

\begin{scope}
\def\hs{0.65}
\node[align=right,ct1] at ($(2)+(-1.1*\hs,0)$) {$d_1\! =\! 7$};
\node[align=left,ct2] at ($(3)+(1.1*\hs,0)$) {$d_2\! =\! 5$};
\node[align=right,ct3] at ($(4)+(-1.1*\hs,0)$) {$d_3\! =\! 3$};
\node[align=right,ct4] at ($(5)+(0,0.8*\hs)$) {$d_4\! =\! 2$};
\end{scope}

\begin{scope}[arcs]
\draw  (9) --  (2) node[pos=0.7,above] {$4.2$};
\draw  (6) --  (2) node[midway,below left=-3pt] {$2.8$};
\draw  (7) --  (3) node[pos=0.5,right=-2pt] {$1.25$};
\draw  (6) --  (3) node[pos=0.5,above=-1pt] {$1.25$};
\draw  (1) --  (3) node[midway,above=2pt] {$2.5$};
\draw  (1) --  (5) node[pos=0.4,above=1pt] {$1.2$};
\draw  (1) --  (6) node[midway,above=3pt] {$5.45$};
\draw  (1) --  (7) node[pos=0.4,above=1pt] {$5.65$};
\draw  (1) -- (10) node[midway,left=-1pt] {$2.2$};
\draw  (7) --  (8) node[midway,above] {$2.9$};
\draw  (8) --  (9) node[pos=0.6,left=-1pt] {$1.4$};
\draw  (6) --  (9) node[midway,above] {$1.4$};
\draw  (7) -- (10) node[midway,right=1pt] {$1.5$};
\draw  (8) --  (4) node[midway,right=-1pt] {$1.5$};
\draw (10) --  (5) node[midway,above] {$0.8$};
\draw (10) -- (11) node[midway,above=-1pt] {$2.9$};
\draw (11) --  (4) node[pos=0.4,below=2pt] {$1.5$};
\draw (11) --  (9) node[pos=0.4,left=-2pt] {$1.4$};
\end{scope}
\end{scope}

\begin{scope}[xshift=7.5cm]

\begin{scope}[every node/.style={vertices},scale=0.5]

\begin{scope}[every node/.append style={ss}]
\node  (1) at (7.16,-9.04) {$s$};
\node[text=ct1]  (2) at (6.60,-15.00) {$t_1$};
\node[text=ct2]  (3) at (10.66,-11.74) {$t_2$};
\node[text=ct3]  (4) at (13.06,-6.56) {$t_3$};
\node[text=ct4]  (5) at (3.50,-6.30) {$t_4$};
\end{scope}

\node  (6) at (3.68,-11.70) {};
\node  (7) at (10.58,-7.66) {};
\node  (8) at (13.30,-9.00) {};
\node  (9) at (14.40,-15.00) {};
\node (10) at (8.00,-4.44) {};
\node (11) at (14.92,-5.38) {};
\end{scope}

\begin{scope}
\def\hs{0.65}
\node[align=right,ct1] at ($(2)+(-1.1*\hs,0)$) {$d_1\! =\! 7$};
\node[align=left,ct2] at ($(3)+(1.1*\hs,0)$) {$d_2\! =\! 5$};
\node[align=right,ct3] at ($(4)+(-1.1*\hs,0)$) {$d_3\! =\! 3$};
\node[align=right,ct4] at ($(5)+(0,0.8*\hs)$) {$d_4\! =\! 2$};
\end{scope}

\begin{scope}[arcs]

\begin{scope}[ct1]
\begin{scope}[p1]
\draw (1) to[out=180, in=90] node[above left=-4pt] {$P_1$} (6);
\draw (6) to[out=-60, in=140] (2);
\end{scope}

\begin{scope}[p2]
\draw (1) to[out=-160,in=60] node[below=-1pt] {$P_2$} (6);
\draw (6) to[out=-20,in=135] (9);
\draw (9) to[out=160,in=40] (2);
\end{scope}

\begin{scope}[p3]
\draw (1) to[bend right=20] node[below=-2pt] {$P_6$} (7);
\draw (7) to[bend right=20] (8);
\draw (8) to[bend right=0] (9);
\draw (9) to[out=170,in=10] (2);
\end{scope}

\begin{scope}[p4]
\draw (1) to[bend right=0] node[pos=0.3,left=-2pt] {$P_9$} (10);
\draw (10) to[out=10,in=135] (11);
\draw (11) to[bend right=0] (9);
\draw (9) to[bend left=15] (2);
\end{scope}

\end{scope}

\begin{scope}[ct2]

\begin{scope}[p1]
\draw (1) to[bend right=-20] node[pos=0.2,below=1pt] {$P_3$} (6);
\draw (6) to[bend right=0] (3);
\end{scope}

\begin{scope}[p2]
\draw (1) to[out=-80, in=150] node[pos=0.1,below=1pt] {$P_4$} (3);
\end{scope}

\begin{scope}[p3]
\draw (1) to[out=-50,in=-90] node[below=1pt] {$P_5$} (7);
\draw (7) to[out=-70,in=70] (3);
\end{scope}

\end{scope}

\begin{scope}[ct3]

\begin{scope}[p1]
\draw (1) to[bend right=0] node[above=-1pt] {$P_7$} (7);
\draw (7) to[bend right=0] (8);
\draw (8) to[bend right=0] (4);
\end{scope}

\begin{scope}[p2]
\draw (1) to[out=70,in=140] node[above=-1pt] {$P_8$} (7);
\draw (7) to[bend right=10] (10);
\draw (10) to[bend right=10] (11);
\draw (11) to[bend right=0] (4);
\end{scope}

\end{scope}

\begin{scope}[ct4]

\begin{scope}[p1]
\draw (1) to[out=130,in=-120] node[left=-1pt] {$P_{10}$} (10);
\draw (10) to[bend right=0] (5);
\end{scope}

\begin{scope}[p2]
\draw (1) to[out=160,in=-50] node[below=-1pt] {$P_{11}$} (5);
\end{scope}

\end{scope}

\end{scope}

\begin{scope}[shift={(8.3,-2.5)}]
\def\vs{0.5}
\node[ct1] at (0,0) {$\lambda_1 = \frac{14}{5}$};
\node[ct1] at (0,-\vs) {$\lambda_2 = \frac{7}{5}$};
\node[ct2] at (0,-2*\vs) {$\lambda_3 = \frac{5}{4}$};
\node[ct2] at (0,-3*\vs) {$\lambda_4 = \frac{5}{2}$};
\node[ct2] at (0,-4*\vs) {$\lambda_5 = \frac{5}{4}$};
\node[ct1] at (0,-5*\vs) {$\lambda_6 = \frac{7}{5}$};
\node[ct3] at (0,-6*\vs) {$\lambda_7 = \frac{3}{2}$};
\node[ct3] at (0,-7*\vs) {$\lambda_8 = \frac{3}{2}$};
\node[ct1] at (0,-8*\vs) {$\lambda_9 = \frac{7}{5}$};
\node[ct4] at (-0.05,-9*\vs) {$\lambda_{10} = \frac{4}{5}$};
\node[ct4] at (-0.05,-10*\vs) {$\lambda_{11} = \frac{6}{5}$};
\end{scope}

\end{scope}

\end{tikzpicture}
\end{center}
\caption{The left figure shows a PSSUF instance $(G,s,T,d,x)$ with $T=\{t_1,t_2,t_3,t_4\}$, and the $x$-values are shown next to the arcs.
The right figure shows an arc-split graph $H$ of $G$ together with a source-numbered nice $s$-path decomposition $(P_i,\lambda_i)_{i\in [11]}$.}
\label{fig:arc-split_ex}
\end{figure}
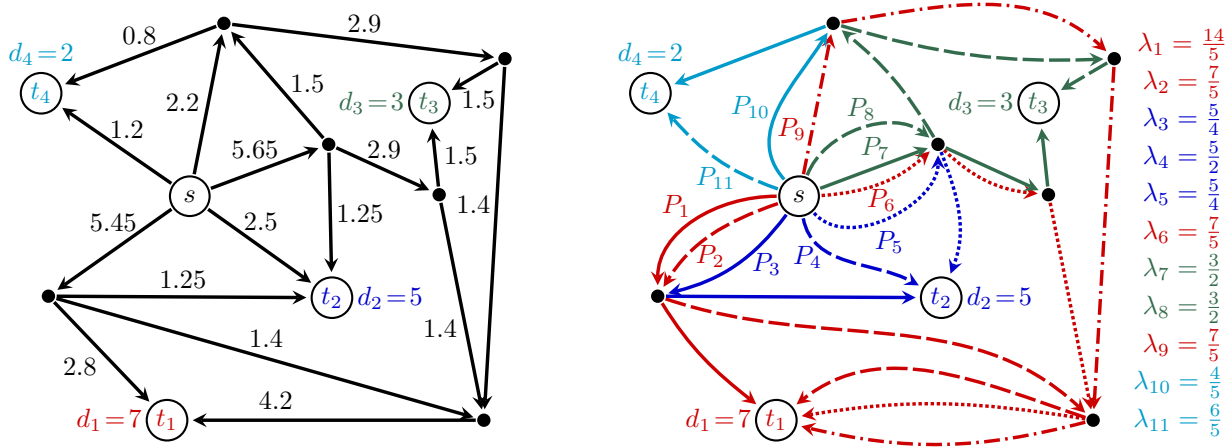

We call paths $P_1,\dots, P_{\ell}$ together with their coefficients $\lambda_1,\dots, \lambda_{\ell}$ fulfilling \cref{item:snNicePaths,item:lambdas} a \emph{source-numbered nice $s$-path decomposition} of $H$.
(For a more formal definition, see \Cref{sec:nice_path_decomp}.)
It can be thought of, as a very structured flow decomposition of $x$.

First, we show that such a decomposition can be obtained in polynomial time.
\begin{theorem}\label{thm:compute_arc_split_overview}
Let $(G=(V,A),s,T,d,x)$ be a PSSUF instance.
We can compute in $O(|V|^2)$ time an arc-split graph $H$ of $G$ together with a source-numbered nice $s$-path decomposition $(P_i,\lambda_i)_{i\in[\ell]}$ of $H$ with $\ell=O(|V|)$.
\end{theorem}

Recall, that in the second step of our algorithm, we will choose our PSSUF solution from the paths $P_1,\dots, P_{\ell}$ (or more formally, the corresponding paths $\phi(P_1),\dots, \phi(P_{\ell})$ in $G$).
More precisely, for each terminal $t\in T$, we will pick one path among
\begin{equation*}
\mathcal{P}^t \coloneqq \{ P_i : i\in[\ell] \text{ and }P_i\text{ ends at terminal }t\}.
\end{equation*}
One important consequence of the paths $(P_i)_{i\in [\ell]}$ being non-crossing, is that the sets $S^t\coloneqq \{ i \in [\ell]: P_i \in \mathcal{P}^t\}$ for $t \in T$ are \emph{non-interleaving}, as we will prove in \Cref{sec:non-interleaving}.

\begin{definition}[Non-interleaving]
Two disjoint sets $S_1,S_2 \subseteq \mathbb{Z}_{>0}$ are \emph{interleaving} if there exist $a_1,b_1 \in S_1$, $a_2,b_2 \in S_2$ with $ a_1 < a_2 <b_1 <b_2$ or $a_2 < a_1 < b_2 < b_1$.
 A partition $\mathcal{S}$ of $[\ell]$ is \emph{non-interleaving} if for any $S_1,S_2 \in \mathcal{S}$ with $S_1 \neq S_2$ the sets $S_1$ and $S_2$ are non-interleaving.
\end{definition}

\begin{restatable}{relem}{nonInterleaving}\label{lem:non-interleaving}
For $t\in T$ let  $S^t \coloneqq\{ i \in [\ell]: P_i \in \mathcal{P}^t \}$.
Then the  partition $\{ S^t : t \in T\}$ of $[\ell]$  is non-interleaving. 
\end{restatable}

This is a key structure we will exploit in our discrepancy-based approach to select one path in each $\mathcal{P}^t$ for each $t\in T$, which we discuss next.

\subsection{Interval-discrepancy on non-interleaving partitions}\label{sec:overview_discrepancy}

We now discuss the selection of paths $P^t\in \mathcal{P}^t$ for all $t\in T$, which we formulate as a discrepancy problem.

\begin{definition}[Selections]
Let $\mathcal{S}$ be a partition of $[\ell]$. Then
\begin{itemize}
\item a \emph{fractional selection} for $\mathcal{S}$ is a vector $y \in \{ q \in [0,1]^\ell \colon q(S) =1 \; \forall S \in \mathcal{S} \}$, and
\item an \emph{integral selection} for $\mathcal{S}$ is a vector $z \in \{ q \in \{0,1\}^\ell \colon q(S) =1 \; \forall S \in \mathcal{S} \}$.
\end{itemize}
\end{definition}

We will apply this definition to the partition $\mathcal{S}$ consisting of the sets  $S^t\coloneqq \{ i \in [\ell]: P_i \in \mathcal{P}^t\}$ with $t \in T$.
The $s$-path decomposition $(P_i, \lambda_i)_{i\in [\ell]}$ naturally gives rise to a fractional selection $y$ for $\mathcal{S}$ by setting $y_i \coloneqq \frac{\lambda_i}{d_t}$ where $t\in T$ is the terminal where the path $P_i$ ends.
Then the vector $(\lambda_i)_{i\in [\ell]}$ will be called the \emph{load vector}.

\begin{definition}[Load vector]
Given a vector $y\in [0,1]^\ell$ and a demand vector $d\in \mathbb{Q}^{\mathcal{S}}_{\ge 0}$, the \emph{load vector} $y^d \in \mathbb{Q}_{\geq 0}^\ell$ is defined by 
\[
y_i^d \coloneqq d_S \cdot y_i,
\]
 where $S \in \mathcal{S}$ is the set containing $i\in [\ell]$.   
\end{definition}

An instance of our discrepancy problem consists of the integer $\ell$, the partition $\mathcal{S}$, the demand vector $d$, and a fractional selection $y$.

\begin{definition}[Weighted partition-constrained selection instance, WPCS instance]
We call a tuple $(\ell,\mathcal{S}, d,y)$ a \emph{weighted partition-constrained selection instance} (for short \emph{WPCS instance}) if
\begin{itemize}
\item $\ell$ is a nonnegative integer,
\item $\mathcal{S}$ is a partition of $[\ell]$,
\item $d$ is a vector in $\mathbb{Q}^{\mathcal{S}}_{\ge 0}$, and
\item $y$ is a fractional selection for $\mathcal{S}$.
\end{itemize}
   \end{definition}

We prove \Cref{thm:main} by finding an integral selection $z$ such that for each arc $a\in A$,
\begin{equation}\label{eq:desired_bound_on_z}
 x(a) - d_{\max} \le \sum_{i\in [\ell]: a\in \phi(P_i) } z^d_i  \le x(a) + d_{\max}.
\end{equation}
Because $(P_i, \lambda_i)_{i\in [\ell]}$ is an $s$-path decomposition, we have 
\[
x(a) = \sum_{i\in [\ell]: a\in \phi(P_i) } \lambda_i = \sum_{i\in [\ell]: a\in \phi(P_i) } y^d_i,
\]
and thus \eqref{eq:desired_bound_on_z} is equivalent to the condition that the \emph{$(y,z)$-discrepancy} of the set $\{i\in [\ell]: a\in \phi(P_i) \}$ is bounded by $d_{\max}$, where the $(y,z)$-discrepancy is defined as follows:

\begin{definition}[Discrepancy]
For a set $I \subseteq [\ell]$ and fractional selections $y$,$z$, we define the \emph{$(y,z)$-discrepancy} of $I$ as 
\[D_{y,z} (I) \coloneqq \abs{y^d(I) - z^d(I)}.\]
\end{definition}

Clearly, in general there is no integral selection $z$ such that the $(y,z)$-discrepancy for every set $I\subseteq [\ell]$ is small. 
(For example, if $y$ is a fractional selection with only small entries, the index set corresponding to the selected elements in any integral selection has very high discrepancy.)
However, we will be able to find an integral selection $z$ such that the $(y,z)$-discrepancy is at most $d_{\max}$ for all \emph{circular intervals}.

\begin{definition}[Circular interval]
For $i,j \in  [\ell]$, the circular interval from $i$ to $j$ is the set 
\begin{itemize}
\item $\{i, i+1, \ldots , j \}$ if $i \leq j$, and
\item $\{i, i+1, \dots, \ell, 1, \dots, j \}$ if $j < i$.
\end{itemize}
A set $I\subseteq [\ell]$ is called a \emph{circular interval} if $I= \emptyset$ or if $I$ is the circular interval from $i$ to $j$
for some $i,j \in  [\ell]$.
\end{definition}

We call the maximum $(y,z)$-discrepancy of any circular interval the \emph{$(y,z)$-interval-discrepancy}.

\begin{definition}[Interval-discrepancy]
For a circular interval $I \subseteq [\ell]$ and fractional selections $y$,$z$, we define the $(y,z)$-interval-discrepancy as 
\[
D_{y,z} \coloneqq \max_{\substack{I \subseteq [\ell] \\ \text{circular interval}}} D_{y,z} (I).
\]
\end{definition}

In \Cref{sec:discrepancy} we will prove the following statement.
\begin{restatable}{rethm}{discrepancyMain}\label{thm:discrepancy_main}
    Consider a WPCS instance $(\ell,\mathcal{S}, d,y)$. 
    If $\mathcal{S}$ is non-interleaving, there exists an integral selection  $z$ with
    \[
    D_{y,z} \leq  d_{\max},
    \]
    where $d_{\max} \coloneqq \max_{S \in \mathcal{S}} d_S$.
    Moreover, it can be computed in time $O(\ell)$. 
\end{restatable}

We will also prove that if we are given costs $c\in \mathbb{Q}^{\ell}$ and require the selection $z$ to be no more expensive than the fractional selection $y$, then we can still achieve an upper bound of $2 d_{\max}$ on the interval discrepancy.
This cost version of the selection problem will be used to prove \Cref{thm:main_cost}.

\begin{restatable}{rethm}{discrepancyMainCosts}\label{thm:discrepancy_main_costs}
    Consider a WPCS instance $(\ell,\mathcal{S}, d,y)$ and let $c \in \mathbb{Q}^{\ell}$.
    If $\mathcal{S}$ is non-interleaving, there exists an integral selection $z$ with
    \begin{align*}
    D_{y,z} &\leq 2 d_{\max} \quad \text{ and}\\
    c^Tz &\leq c^Ty,
    \end{align*}
    where $d_{\max} \coloneqq \max_{S \in \mathcal{S}} d_S$
        Moreover, it can be computed in time $O(\ell \cdot \inp)$, where $\inp$ refers to the length of the bit encoding of the input. 
\end{restatable}

To obtain \Cref{thm:discrepancy_main}, we exploit that $\mathcal{S}$ is non-interleaving to prove that a simple greedy algorithm yields an integral selection with the desired properties.
For showing \Cref{thm:discrepancy_main_costs}, we first use a common argument in discrepancy theory (see, e.g, \cite{lovasz_1986_discrepancy}, and \cite{bansal2022flow} for a recent application in a scheduling context), allowing us to reduce to the special case where each set $S\in \mathcal{S}$ has size two and the given fractional selection $y$ fulfills $y_i=\frac{1}{2}$ for all $i\in[\ell]$; this reduction comes at the cost of losing a factor of two in our discrepancy bound.
In the special case we just described, we can for each integral selection $z$ define a complementary integral selection $\overline{z}$ that selects for each $S\in \mathcal{S}$ the element not selected by $z$.
We observe that whenever $z$ fulfills the desired discrepancy bounds, then so does $\overline{z}$.
Moreover, one of the selections $z$ and $\overline{z}$ is not more expensive than $y$.
Hence, we can apply the same algorithm used to prove \Cref{thm:discrepancy_main} to this half-integral discrepancy problem in order to find an integral selection $z$ fulfilling the desired discrepancy bounds and then return the cheaper selection of $z$ and $\overline{z}$.
For details, see \Cref{sec:discrepancy}.

\subsection{Bounding interval-discrepancy suffices}

Recall that our goal is to find an integral selection $z$ such that the discrepancy of the set $\{i\in [\ell]: a\in \phi(P_i) \}$ is small for every arc $a\in A$.
By \Cref{thm:discrepancy_main} and \Cref{thm:discrepancy_main_costs} we can achieve a small discrepancy for every circular interval.
However, not every set $\{i\in [\ell]: a\in \phi(P_i) \}$ with $a\in A$ is a circular interval, as the example in \Cref{fig:not-interval} shows.

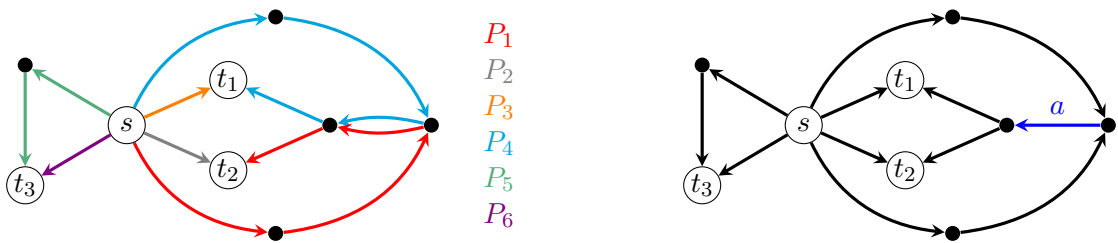
\begin{figure}[htb]
\begin{center}
\begin{tikzpicture}[rotate=-90,xscale=0.8, yscale=0.9,
terminals/.style={circle,draw=black, inner sep=0pt, minimum size=14pt,},
vertices/.style={circle,fill=black,inner sep=2pt},
arcs/.style={line width=1.2pt,-stealth},
p1/.style={red},
p2/.style={gray},
p3/.style={orange},
p4/.style={cyan!90!black},
p5/.style={green!80!black},
p6/.style={violet},
]

\begin{scope}
\node[terminals] (s) at (0,0) {$s$};
\node[terminals] (t1) at (-0.75,1.5) {$t_1$};
\node[terminals] (t2) at (0.75,1.5) {$t_2$};
\node[terminals] (t3) at (1,-1.5) {$t_3$};
\node[vertices] (1) at (-1,-1.5) {};
\node[vertices] (2) at (-1.8,2.2) {};
\node[vertices] (3) at (1.8,2.2) {};
\node[vertices] (4) at (0,4.5) {};
\node[vertices] (5) at (0,3) {};
\end{scope}

\begin{scope}[arcs]
\begin{scope}[p4]
\draw[bend left] (s) to (2);
\draw[bend left] (2) to (4);
\draw[bend right=15] (4) to (5);
\draw (5) -- (t1);
\end{scope}
\draw[p3] (s) to (t1);
\begin{scope}[p1]
\draw[bend right] (s) to (3);
\draw[bend right] (3) to (4);
\draw[bend left=15] (4) to (5);
\draw (5) -- (t2);
\end{scope}
\draw[p2] (s) to (t2);
\begin{scope}[p5]
\draw (s) to (1);
\draw (1) to (t3);
\end{scope}
\draw[p6] (s) to (t3);
\end{scope}

\node[p1] at (-1.5,5.5) {$P_1$};
\node[p2] at (-0.9,5.5) {$P_2$};
\node[p3] at (-0.3,5.5) {$P_3$};
\node[p4] at (0.3,5.5) {$P_4$};
\node[p5] at (0.9,5.5) {$P_5$};
\node[p6] at (1.5,5.5) {$P_6$};

\begin{scope}[shift={(0,10)}]
\node[terminals] (s) at (0,0) {$s$};
\node[terminals] (t1) at (-0.75,1.5) {$t_1$};
\node[terminals] (t2) at (0.75,1.5) {$t_2$};
\node[terminals] (t3) at (1,-1.5) {$t_3$};
\node[vertices] (1) at (-1,-1.5) {};
\node[vertices] (2) at (-1.8,2.2) {};
\node[vertices] (3) at (1.8,2.2) {};
\node[vertices] (4) at (0,4.5) {};
\node[vertices] (5) at (0,3) {};

\begin{scope}[arcs]
\draw[bend left] (s) to (2);
\draw[bend left] (2) to (4);
\draw (5) -- (t1);
\draw (s) to (t1);
\draw[bend right] (s) to (3);
\draw[bend right] (3) to (4);
\draw (5) -- (t2);
\draw (s) to (t2);
\draw (s) to (1);
\draw (1) to (t3);
\draw (s) to (t3);
\draw[blue] (4) -- node[above] {$a$} (5);
\end{scope}
\end{scope}

\end{tikzpicture}

 \end{center}
\caption{\label{fig:not-interval}
The figure shows an example of an arc-split graph $H$ (left) with a family of paths $P_1,\dots, P_6$ from a source-numbered nice $s$-path decomposition.
The graph $G$ is shown on the right with an arc $a$ highlighted in blue.
In this example we have $\{i\in [\ell]: a\in \phi(P_i) \} = \{1,4\}$, which is not a circular interval.
}
\end{figure}

Nevertheless, we will prove that \Cref{thm:discrepancy_main} and \Cref{thm:discrepancy_main_costs} can be used to bound the discrepancy of all sets $\{i\in [\ell]: a\in \phi(P_i) \}$ with $a\in A$.
In order to prove this, we define an equivalence relation on the set $2^{[\ell]}$ of all subsets of $[\ell]$.
Recall that $S^t \coloneqq \{ i\in [\ell] : P_i\text{ ends at terminal }t\}$.

\begin{definition}[$S^t$-addition and $S^t$-removal]
Let $t\in T$.
We say that a set $Y\subseteq [\ell]$ \emph{arises from }$X\subseteq [\ell]$\emph{ by $S^t$-addition}, if $X\cap S^t = \emptyset$ and $Y=X\cup S^t$. 
We say that a set $Y\subseteq [\ell]$ \emph{arises from }$X\subseteq [\ell]$\emph{ by $S^t$-deletion}, if $S^t \subseteq X$ and $Y=X\setminus S^t$.
\end{definition}

If a set $Y\subseteq [\ell]$ arises from $X \subseteq [\ell]$ by (potentially several) $S^t$-additions or $S^t$-removals (for some terminals $t\in T$),  the $(y,z)$-discrepancy of $X$ and $Y$ is identical for every integral selection $z$.
For this reason we will then say that $X$ and $Y$ are \emph{discrepancy-equivalent}.

\begin{definition}[Discrepancy-equivalent]
Two sets $X, Y \subseteq [\ell]$ are \emph{discrepancy-equivalent} if $Y$ arises from $X$ by a sequence of $S^t$-additions and $S^t$-removals for some terminals $t\in T$.
\end{definition}

Note that this indeed defines an equivalence relation.
We will prove in \Cref{sec:discr_equiv} the statement below, which says that the set of indices of all paths in $\mathcal{P}$ containing a copy of the same given arc $a\in A$ are discrepancy-equivalent to a circular interval.
For example, the indices of the paths in \Cref{fig:not-interval} that go over (a copy of) arc $a$, which are $1$ and $4$, are discrepancy-equivalent to the circular interval $\{4,5,6,1\}$, which is obtained from $\{1,4\}$ by $S^{t_3}$-addition.
\begin{restatable}{relem}{lemDiscrEquivToInterval}\label{lem:discr_equiv_to_interval}
For every arc $a \in A$, the set $\{  i \in [\ell]: a \in \phi(P_i)  \}$ is discrepancy-equivalent to a circular interval.
\end{restatable}

The proof of this statement (in \Cref{sec:properties_path_decomp}) crucially exploits that $P_1,\dots, P_{\ell}$ are paths from a source-numbered nice $s$-path decomposition of the arc-split graph $H$.
Using \Cref{lem:non-interleaving} and \Cref{lem:discr_equiv_to_interval}, we can now prove that our discrepancy statements (\Cref{thm:discrepancy_main} and \Cref{thm:discrepancy_main_costs}) imply our main results (\Cref{thm:main} and \Cref{thm:main_cost}).

\begin{proof}[Proof of  \Cref{thm:main} and \Cref{thm:main_cost}]
Let $(G=(V,A),s,T,d,x)$ be a PSSUF instance. 
By \Cref{thm:compute_arc_split} we can compute in $O(|V|^2)$ time an arc-split $(H,\pi,\phi)$ of $G$ and a nice $s$-path decomposition $(P_i,\lambda_i)_{i\in[\ell]}$ of $H$ with $\ell=O(|V|)$.
We define a WPCS instance $(\ell, \mathcal{S}, d,y)$ as follows.
For $t\in T$ we let  $S^t \coloneqq \{ i\in [\ell]: P_i\text{ is an $s$-$t$ path}\}$ and we define $\mathcal{S}\coloneqq \{S^t : t\in T\}$.
Because $G$ is acyclic and $\sum_{i=1}^\ell \lambda_i \chi^{\phi(P_i)} = x$, every path $P_i$ with $i\in [\ell]$ starts at $s$ and ends at some terminal $t\in T$.
Therefore, $\mathcal{S}$ is a partition of $[\ell]$.
We define the demand vector $d\in \mathbb{Q}_{\ge 0}^{\mathcal{S}}$ by $d_{S^t} \coloneqq d_t$.
For $i\in [\ell]$ we define $y_i \coloneqq \frac{\lambda_i}{d_t}$, where $t\in T$ is the terminal where the path $P_i$ ends.
Because $\sum_{i=1}^\ell \lambda_i \chi^{\phi(P_i)} = x$ and $x(\delta^-(t)) - x(\delta^+(t)) = d_t$ for all $t\in T$, we have 
$\sum_{i\in S^t} \lambda_i = d_t$ and thus $\sum_{i\in S^t} y_i = 1$ for all $t\in T$.
This shows that $y$ is a fractional selection for $\mathcal{S}$.
By \Cref{lem:non-interleaving}, the partition $\mathcal{S}$ is non-interleaving.

Now consider an integral selection $z$ for $\mathcal{S}$.
For each $t\in T$, we select the $s$-$t$ path $P^t \coloneqq P_i$, where $i$ is the unique element of $S^t$ with $z_i =1$.
Consider the unsplittable flow $\mathcal{P}=\{P^t\}_{t\in T}$ and an arc $a\in A$.
By \Cref{lem:discr_equiv_to_interval}, the set $\{ i \in [\ell] : a\in \phi(P_i) \}$ is discrepancy equivalent to some circular interval $I_a$, implying $D_{y,z}(\{i\in [\ell]: a\in \phi(P_i)\}) = D_{y,z}(I_a) $.
Using 
\[
x(a) = \sum_{i\in [\ell]: a\in \phi(P_i)} \lambda_i =  \sum_{i\in [\ell]: a\in \phi(P_i)} y^d_i, 
\]
and
\[
  \flow_{\mathcal{P}}(a)  = \sum_{i\in [\ell]: a\in \phi(P_i)} z^d_i, 
\]
this implies 
\begin{equation*}
\abs{x(a) - \flow_{\mathcal{P}}(a) } = \abs{ \sum_{i\in [\ell]: a\in \phi(P_i)} y^d_i  - \sum_{i\in [\ell]: a\in \phi(P_i)} z^d_i } \\
= D_{y,z}(\{i\in [\ell]: a\in \phi(P_i)\}) \\
= D_{y,z}(I_a)  \\
\le D_{y,z}.
\end{equation*}

We conclude that \Cref{thm:main} follows from the discrepancy statement \Cref{thm:discrepancy_main} applied to the WPCS instance $(\ell, \mathcal{S}, d,y)$. As an overall runtime we obtain $O(\abs{V}^2)$.
Moreover, \Cref{thm:main_cost} follows from \Cref{thm:discrepancy_main_costs} applied to the WPCS instance $(\ell, \mathcal{S}, d,y)$ and the cost vector $c\in \mathbb{Q}_{\ge 0}^{\ell}$ defined by $c_i \coloneqq c(P_i) \coloneqq \sum_{a\in P_i} c(a)$. As an overall runtime we obtain $O(\abs{V}^2 + \inp \cdot \abs{V})$.
\end{proof}

The remainder of this paper is structured as follows.
In \Cref{sec:nice_path_decomp} we show how to obtain an arc split graph $H$ with a nice $s$-path decomposition and prove \Cref{thm:compute_arc_split}.
In \Cref{sec:properties_path_decomp}, we prove that a nice $s$-path decomposition has the properties claimed in \Cref{lem:non-interleaving} and \Cref{lem:discr_equiv_to_interval}.
\Cref{sec:discrepancy} contains the proof of the two discrepancy statements \Cref{thm:discrepancy_main} and \Cref{thm:discrepancy_main_costs}.
\Cref{sec:conclusions} contains some concluding remarks.
Finally, \Cref{sec:hardness_cyclic_graphs} discusses the necessity of the acyclicity assumption when dealing with lower bounds.

\section{Computing a Nice Path Decomposition}\label{sec:nice_path_decomp}

We now provide details on how we compute a source-numbered nice path decomposition of a flow $x$, which proves \Cref{thm:compute_arc_split_overview}.
Moreover, we expand on the discussion of \Cref{sec:overview_path_decomp} and further formalize some of the concepts introduced there.

For a vertex $v\in V$, let $\delta(v) \coloneqq \delta^+(v) \cup \delta^-(v)$ be the set of all arcs incident to $v$.
Each planar geometric embedding induces a cyclic ordering $\pi_v$ of the arcs in $\delta(v)$, obtained by traversing the arcs in $\delta(v)$ in counterclockwise sense.
We say that $b\in  \delta(v)$ is the \emph{successor} of $a\in\delta(v)$ in $\pi_v$ if, in the geometric embedding, $b$ is the next arc after $a$ in counterclockwise direction.
In this case, we also say that $a$ is the \emph{predecessor} of $b$ in $\pi_v$.

A collection $\pi=(\pi_v)_ {v\in V}$ of cyclic orderings  of the arcs incident to $v$ for all $v\in V$ is called a \emph{combinatorial embedding}. 
A combinatorial embedding is \emph{planar} if it is induced by some planar geometric embedding.
See \Cref{fig:combinatorial_embedding} for an example.
\begin{figure}[htb]
\begin{center}
\begin{tikzpicture}[yscale=0.7,
vertices/.style={circle,draw=black,inner sep=2pt},
arcs/.style={line width=1.5pt,-stealth},
]

\begin{scope}[every node/.style={vertices}]
\node (v) at (0,0) {$v$};
\end{scope}

\begin{scope}
\coordinate (1) at (-1.5, 1);
\coordinate (2) at (-1.5,-1);
\coordinate (3) at (0,-2);
\coordinate (4) at ( 1.5,-1);
\coordinate (5) at ( 1.5, 1);
\coordinate (6) at (0,2);
\end{scope}

\begin{scope}[arcs]
\draw (v) -- node[above] {$a_1$} (1);
\draw (2) -- node[above] {$a_2$} (v);
\draw (v) -- node[left] {$a_3$} (3);
\draw (4) -- node[below] {$a_4$} (v);
\draw (5) -- node[below] {$a_5$} (v);
\draw (v) -- node[right] {$a_6$} (6);
\end{scope}

\end{tikzpicture}
 \end{center}
\caption{\label{fig:combinatorial_embedding}
A vertex $v$ with incident arcs $a_1, \dots a_6$, geometrically embedded into the plane. 
Let $\pi$ be the induced combinatorial embedding.
In the  cyclic ordering $\pi_v$, the arc $a_6$ is the predecessor of $a_1$ and the arc $a_2$ is the successor of $a_1$.\newline
The tuple $(a_5,a_2,a_3)$ is a $\pi_v$-progression because $a_5,a_2,a_3$ is a subsequence of $a_5,a_6,a_1,a_2,a_3,a_4$.
The tuple $(a_5,a_3,a_2)$ is not a $\pi_v$-progression.
}
\end{figure}
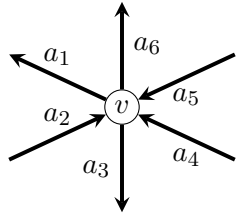

In order to formally define when two paths are crossing, we introduce the notion of a $\pi_v$-progression.
\Cref{fig:combinatorial_embedding} shows an example.

\begin{definition}[$\pi_v$-progression]
Let $v\in V$ and let $b_1,\dots, b_k\in \delta(v)$ be $k$ distinct arcs.
Let $\delta(v)=\{a_1,\dots, a_m\}$, where the numbering is chosen such that $a_1=b_1$ and $a_{i+1}$ is the successor of $a_i$ in the cyclic order $\pi_v$ for all $i\in [m-1]$.
Then $(b_1,b_2,\dots, b_k)$ is a \emph{$\pi_v$-progression} if $b_1,b_2,\dots,b_k$ is a subsequence of $a_1,\dots,a_m$.
\end{definition}

We are now ready to formally define when two edge-disjoint paths are crossing.
See \Cref{fig:crossing_paths}.

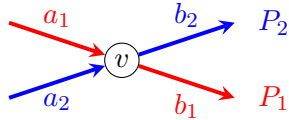
\begin{figure}[htb]
\begin{center}

\begin{tikzpicture}[
vertices/.style={circle,draw=black,inner sep=2pt},
arcs/.style={line width=1.5pt,-stealth},
p1/.style={red},
p2/.style={blue}
]

\begin{scope}[every node/.style={vertices}]
\node (v) at (0,0) {$v$};
\end{scope}

\begin{scope}
\coordinate (u1) at (-1.5, 0.5);
\coordinate (u2) at (-1.5,-0.5);
\coordinate (w1) at ( 1.5,-0.5);
\coordinate (w2) at ( 1.5, 0.5);
\end{scope}

\begin{scope}[arcs]
\begin{scope}[p1]
\draw (u1) -- node[above] {$a_1$} (v);
\draw (v) -- node[below] {$b_1$} (w1);
\end{scope}

\begin{scope}[p2]
\draw (u2) -- node[below] {$a_2$} (v);
\draw (v) -- node[above] {$b_2$} (w2);
\end{scope}
\end{scope}

\begin{scope}
\node[p1] at (w1)[right=5pt] {$P_1$};
\node[p2] at (w2)[right=5pt] {$P_2$};
\end{scope}

\end{tikzpicture}
 \end{center}
\caption{$P_1$ and $P_2$ are crossing because $(a_1,a_2,b_1,b_2)$ is a $\pi_v$-progression.}
    \label{fig:crossing_paths}
\end{figure}

\begin{definition}[Crossing paths]
\label{def:crossing_paths}
Let $H=(V,F)$ be a graph with planar combinatorial embedding $\pi$.
Two arc-disjoint paths $P_1, P_2$ are \emph{crossing} if they have a common mutually internal vertex $v\in V$ such that $(a_1,a_2,b_1,b_2)$ or $(a_2,a_1,b_2,b_1)$ is a $\pi_v$-progression, where $\{a_1\} = P_1\cap \delta^-(v)$, $\{b_1\} = P_1\cap \delta^+(v)$, $\{a_2\} = P_2\cap \delta^-(v)$, and $\{b_2\} = P_2\cap \delta^+(v)$. 
\end{definition}

This notion works well for arc-disjoint paths but fails to capture what would naturally be considered crossing paths if the paths have arcs in common, as illustrated in \Cref{fig:problem_without_arc_split}.
For this reason, we work with the arc-split graph $H$ of $G$, where arcs are replaced by (a well-chosen arc-dependent number of) parallel copies.
This allows for assuming that the paths $P_1,\dots, P_{\ell}$ we construct later are disjoint in $H$.
\begin{figure}[htb]
\begin{center}
\begin{tikzpicture}[
vertices/.style={circle,draw=black,inner sep=2pt,minimum size=19},
arcs/.style={line width=1.5pt,-stealth},
p1/.style={blue},
p2/.style={red},
p3/.style={green!80!black},
p4/.style={orange!80!black}
]

\begin{scope}
\node[vertices] (v1) at (-1,0) {$v_1$};
\node[vertices] (w1) at (1,0) {$w_1$};
\coordinate (u1) at (-2.5, 0.5);
\coordinate (x1) at (-2.5,-0.5);
\coordinate (r1) at ( 2.5,-0.5);
\coordinate (y1) at ( 2.5, 0.5);
\end{scope}

\begin{scope}[shift={(0,-2)}]
\node[vertices] (v2) at (-1,0) {$v_2$};
\node[vertices] (w2) at (1,0) {$w_2$};
\coordinate (u2) at (-2.5, 0.5);
\coordinate (x2) at (-2.5,-0.5);
\coordinate (r2) at ( 2.5,-0.5);
\coordinate (y2) at ( 2.5, 0.5);
\end{scope}

\begin{scope}[shift={(8,0)}]
\node[vertices] (v3) at (-1,0) {$v_1$};
\node[vertices] (w3) at (1,0) {$w_1$};
\coordinate (u3) at (-2.5, 0.5);
\coordinate (x3) at (-2.5,-0.5);
\coordinate (r3) at ( 2.5,-0.5);
\coordinate (y3) at ( 2.5, 0.5);
\end{scope}

\begin{scope}[shift={(8,-2)}]
\node[vertices] (v4) at (-1,0) {$v_2$};
\node[vertices] (w4) at (1,0) {$w_2$};
\coordinate (u4) at (-2.5, 0.5);
\coordinate (x4) at (-2.5,-0.5);
\coordinate (r4) at ( 2.5,-0.5);
\coordinate (y4) at ( 2.5, 0.5);
\end{scope}

\begin{scope}[arcs]

\draw (v1) -- (w1);

\begin{scope}[p1]
\draw (u1) --  (v1);
\draw (w1) --  (r1);
\draw (u3) --  (v3);
\draw[bend left] (v3) to (w3);
\draw (w3) --  (r3);
\end{scope}
\begin{scope}[p2]
\draw (x1) --  (v1);
\draw (w1) -- (y1);
\draw (x3) --  (v3);
\draw[bend right] (v3) to (w3);
\draw (w3) -- (y3);
\end{scope}

\draw (v2) -- (w2);
\begin{scope}[p4]
\draw (u2) -- (v2);
\draw (w2) -- (y2);
\draw (u4) -- (v4);
\draw[bend left] (v4) to (w4);
\draw (w4) -- (y4);
\end{scope}
\begin{scope}[p3]
\draw (x2) --  (v2);
\draw (w2) -- (r2);
\draw (x4) --  (v4);
\draw[bend right] (v4) to (w4);
\draw (w4) -- (r4);
\end{scope}

\end{scope}

\begin{scope}
\node[p1] at (r1)[right=5pt] {$P_2$};
\node[p2] at (y1)[right=5pt] {$P_1$};
\node[p3] at (r2)[right=5pt] {$P_4$};
\node[p4] at (y2)[right=5pt] {$P_3$};
\node[p1] at (r3)[right=5pt] {$P_2$};
\node[p2] at (y3)[right=5pt] {$P_1$};
\node[p3] at (r4)[right=5pt] {$P_4$};
\node[p4] at (y4)[right=5pt] {$P_3$};
\end{scope}

\end{tikzpicture}

 \end{center}
\caption{\label{fig:problem_without_arc_split}
The left-hand side shows paths $P_1$ and $P_2$ that use a common arc $(v_1,w_1)$ and paths $P_3$ and $P_4$ that use a common arc $(v_2,w_2)$.
We would like to consider $P_1$ and $P_2$ as crossing, but not $P_3$ and $P_4$.
However, when considering the local cyclic arc-orderings at vertices, we cannot distinguish these two cases.
This is why we introduce arc-split graphs, where $(v_i,w_i)$ is replaced by two copies, each used by one path, as depicted on the right-hand side.
This leads to arc-disjoint paths $P_1$ and $P_2$ that are crossing according to \Cref{def:crossing_paths}.
(In the example, they cross at $w_1$, and if we swap the red and blue copy of the arc $(v_1,w_1)$, they cross at $v_1$.)
The paths $P_3$ and $P_4$ are not crossing.
}
\end{figure}
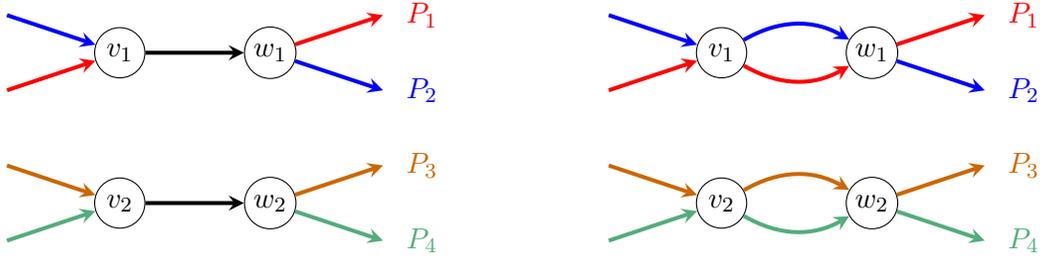

Arc-split graphs are formally defined as follows, where for some arc $a=(v,w)$, we write $\head(a)=w$ and $\tail(a)=v$.

\begin{definition}[Arc-split graph $H$ of $G$]
\label{def:arc_split_graph}
Let $G=(V,A)$ be a planar graph.
An \emph{arc-split} graph of $G$ is a tuple $(H,\pi)$, where $H=(V,F)$ is a planar graph with planar combinatorial embedding $\pi$, such that there exists a mapping $\phi \colon F \to A$ fulfilling for any $f\in F$:
\begin{equation*}
\head(f) =\head(\phi(f)) \quad\text{ and }\quad \tail(f) =\tail(\phi(f)),
\end{equation*}
and, for all $a\in A$, the set $\phi^{-1}(a) \subseteq F$ is consecutive at both $\head(a)$ and $\tail(a)$ in $\pi$.
We call $(H,\pi,\phi)$ an \emph{arc-split} of $G$.
\end{definition}

As mentioned above, we want the paths $P_1,\dots, P_{\ell}$ to be arc-disjoint in the arc split graph $H$.
The number of copies of an arc $a\in A$ in the arc-split graph $H$ that we construct will be chosen such that each copy is used by some path $P_i$. 
Thus, the paths $P_1,\dots, P_{\ell}$ that we construct will form a \emph{nice $s$-path partitioning}  of $H$, which is defined as follows.

\begin{definition}[Nice $s$-path partitioning]
Let $H=(V,F)$ be a graph with planar combinatorial embedding $\pi$, and let $s\in V$.
A collection of paths $P_1,\ldots, P_\ell \subseteq F$ in $H$ is called an \emph{$s$-path partitioning} of $H$ if
\begin{itemize}
    \item $P_i$ starts at $s$ for all $i \in [\ell]$, and
    \item $F = P_1 \dot\cup \ldots \dot\cup P_{\ell}$.
\end{itemize}
It is called \emph{nice} if the paths $P_1, \ldots , P_{\ell}$ are pairwise non-crossing.
\end{definition}

Together with the nice $s$-path partitioning of $H$ we will also compute non-negative weights for the paths $P_i$ in order to obtain a path decomposition of the flow $x$.

\begin{definition}[Nice $s$-path decomposition]
Let $(G,s,T,d,x)$ be a PSSUF instance and let $(H,\pi,\phi)$ be an arc-split of $G=(V,A)$.
A \emph{nice $s$-path decomposition} of $H$ is a sequence of tuples $(P_i,\lambda_i)_{i\in [\ell]}$, where $P_1,\dots, P_\ell$ is a nice $s$-path partitioning of $H$ and $\lambda_i\in \mathbb{Q}_{\geq 0}$ for $i\in [\ell]$, such that
\begin{equation*}
\sum_{i=1}^\ell \lambda_i \chi^{\phi(P_i)} = x.\footnotemark
\end{equation*}
\end{definition}
\footnotetext{$\phi(P_i)\coloneqq \{\phi(f) \colon f\in P_i\}$ denotes the image of $P_i$ by $\phi$.
Moreover, for a set $U\subseteq A$, we denote by $\chi^U\in \{0,1\}^A$ the characteristic vector of $U$, i.e., the zero-one vector in $\{0,1\}^A$ with a $1$ precisely in the entries corresponding to arcs in $U$.
}

Recall that the paths $P_1,\dots, P_{\ell}$ are all paths starting in $s$.
Thus, the cyclic ordering $\pi_s$ of the arcs incident to $s$ naturally induces a cyclic ordering on the paths $P_1, \dots, P_{\ell}$.
As mentioned in \Cref{sec:overview}, if the numbering of these paths is consistent with this cyclic ordering, we say that $P_1, \dots, P_{\ell}$ are \emph{source-numbered}.

\begin{definition}[Source-numbered]
An $s$-path partitioning is called \emph{source-numbered} if $(a_1, \ldots, a_{\ell})$ is a $\pi_s$-progression, where $a_i$ is the first arc of $P_i$ for $i \in [\ell]$.
\end{definition}

We now prove \Cref{thm:compute_arc_split_overview}, which we restate here using the just introduced notation, making explicit that we compute a combinatorial embedding of $H$.
\begin{restatable}{rethm}{computeArcSplit}\label{thm:compute_arc_split}
Let $(G=(V,A),s,T,d,x)$ be a PSSUF instance.
We can compute in $O(|V|^2)$ time an arc-split $(H,\pi,\phi)$ of $G$ and a source-numbered nice $s$-path decomposition $(P_i,\lambda_i)_{i\in[\ell]}$ of $H$ with $\ell=O(|V|)$.
\end{restatable}

\subsection{Proof of \Cref{thm:compute_arc_split}}
Note that source-numbering a set of paths starting at $s$ simply means that the paths are numbered in counterclockwise sense around $s$ (starting at an arbitrary path).
As this is trivial to do at the end, we will ignore this aspect when constructing a source-numbered nice $s$-path decomposition fulfilling the conditions of \Cref{thm:compute_arc_split}, and thus just focus on finding a nice $s$-path decomposition.

There are different ways to obtain nice $s$-path decompositions.
One option is to start with an arbitrary flow decomposition and then iteratively reduce the number of crossings by uncrossing steps.
We present another approach that avoids crossings upfront and leads to a running time that is linear in the total size $\sum_{i=1}^{\ell} |P_i|$ of the computed paths.

To obtain the desired nice $s$-path decomposition in an arc-split graph $H$ of $G$, as claimed by \Cref{thm:compute_arc_split}, we first decide locally on rules how paths can traverse each vertex $v\in V\setminus (T\cup \{s\})$.
These rules are set up such that paths that we later construct from them will be non-crossing.
We formalize these rules through the notion of a \emph{good $v$-wiring}.
A good $v$-wiring is a sequence of pairs, each coupling one $v$-incoming arc $f\in \delta^-(v)$ with one $v$-outgoing arc $g\in \delta^+(v)$, together with a maximum load $\mu_v((f,g))$ with which we can use this pair.
The paths in our nice $s$-path decomposition that go over $v$ will use one of these pairs, and the load with which each pair $(f,g)\in \delta^-(v)\times \delta^+(v)$ is used is given by $\mu_v((f,g))$.
\begin{definition}[Good $v$-wiring]\label{def:good_v_wiring}
Let $(G=(V,A),s,T,d,x)$ be a PSSUF instance, and let $v\in V\setminus (T\cup \{s\})$.
Moreover, let $\pi^G$ be a planar combinatorial embedding of $G$.
A \emph{good $v$-wiring} is a tuple $(W_v,\mu_v)$, where $W_v=((f_1,g_1),\dots,(f_k,g_k))$ is a sequence of distinct arc pairs $(f_i,g_i) \in \delta^-(v) \times \delta^+(v)$ for $i\in [k]$, and $\mu_v\colon W_v \to \mathbb{Q}_{\geq 0}$, such that
\begin{equation*}
\displaystyle\sum_{\substack{i\in [k]:\\a\in \{f_i,g_i\}}} \mu_v((f_i,g_i)) = x(a) \qquad \forall a\in \delta(v),
\end{equation*}
and, for $i,j\in [k]$ with $i < j$, we have that neither $(f_i,f_j,g_i)$ nor $(f_i,g_j,g_i)$ is a $\pi^G_v$-progression.
\end{definition}
Note that forbidding, for $i,j\in [k]$ with $i<j$, the $\pi^G_v$-progressions $(f_i,f_j,g_i)$ and $(f_i,g_j,g_i)$, excludes the $\pi^G_v$-progressions $(f_i,f_j,g_i,g_j)$ and $(f_j,f_i,g_j,f_i)$.

We start by showing that a good $v$-wiring can be computed fast.

\begin{lemma}\label{lem:comp_good_v_wiring}
Let $(G=(V,A),s,T,d,x)$ be a PSSUF instance with planar combinatorial embedding $\pi^G$, and let $v\in V\setminus (T\cup \{s\})$.
We can find in $O(\deg(v))$ time a good $v$-wiring $(W_v, \mu_v)$ with $|W_v|\leq \deg(v)$.
\end{lemma}
\begin{proof}
We first describe a simple procedure to obtain the desired result, which, when implemented in a straightforward way, only runs in polynomial time instead of the desired $O(\deg(v))$ running time bound.
In a second step, we expand on how to refine this procedure to obtain the claimed linear time bound.

We construct a good $v$-wiring by iteratively finding a pair of arcs $a\in \delta^-(v)$, $b\in \delta^+(v)$ with $b$ being $a$'s successor in $\pi^G_v$, adding the pair $(a,b)$ to $W_v$ with a $\mu_v$-value equal to $\min\{x(a),x(b)\}$, reducing the $x$-values of $a$ and $b$ correspondingly, and deleting $a$ and/or $b$ if their $x$-value reached $0$.
\Cref{alg:good_v_wiring_simple} formalizes this procedure.
In its implementation, we use a circular doubly linked list $L$ of the arcs in $\delta(v)$.
This list contains all arcs in $\delta(v)$ in the order given by $\pi^G_v$.
For an arc $a\in L$, we denote by $\succelem_L(a)$ the arc in $L$ that comes after $a$ with respect to the circular ordering $\pi^G_v$, i.e., the one after $a$ in $L$ in counterclockwise order.
\LinesNumbered
\begin{algorithm2e}[ht]
\DontPrintSemicolon
\caption{Simple algorithm to create good $v$-wiring.}
\label{alg:good_v_wiring_simple}
Let $y=x$.\;
Let $L$ be a circular doubly linked list of $\delta(v)$.\;
$W_v = ()$.\;
\While {$L\neq \emptyset$}{
Find an arc $a\in L \cap \delta^-(v)$ such that $b= \succelem_L(a)\in \delta^+(v)$.\label{algline:find_ab_pair_simple}\;
Add $(a,b)$ to $W_v$.\;
Set $\mu_v((a,b))=\min\{y(a),y(b)\}$.\;
$y(a) = y(a) - \mu_v((a,b))$.\;
$y(b) = y(b) - \mu_v((a,b))$.\;
\lIf {$y(a)=0$}{
delete $a$ from $L$.\label{algline:delete_a_simple}
}
\lIf {$y(b)=0$}{
delete $b$ from $L$.\label{algline:delete_b_simple}
}
}
\Return $(W_v,\mu_v)$.\;
\end{algorithm2e}

Note that we have $y(\delta^+(v))=y(\delta^-(v))$ during the algorithm.
At the start, this holds because $x(\delta^+(v))=x(\delta^-(v))$, and the property is preserved because, whenever a pair gets added to $W_v$, we remove the same value from one incoming and one outgoing arc of $v$.
This implies that, indeed, in every iteration of the algorithm there exists a pair $(a,b)$ of arcs as desired.

Also observe that this algorithm can clearly be implemented to have polynomial running time.
Indeed, at every iteration of the while loop, at least one of the arcs $a$ or $b$ will be removed from $L$, leading to at most $O(\deg(v))$ many iterations.
Moreover, the bottleneck operation in one iteration is finding an arc $a$ in line~\ref{algline:find_ab_pair_simple}.
A trivial implementation of this step, where we simply check all arcs in $L$, leads to a running time of $O(\deg(v)^2)$.

Before expanding on how to obtain a faster implementation, we show that $(W_v,\mu_v)$ is a good $v$-wiring.
Let $W_v=((f_1,g_1),\dots,(f_k,g_k))$.
Note that we clearly have
\begin{equation*}
\sum_{\substack{i\in [k]:\\ a\in \{f_i,g_i\}}} \mu_v((f_i,g_i)) = x(a) \qquad \forall a\in \delta(v),
\end{equation*}
because the procedure stops once all $y$-values, which are the leftover $x$-values, are set to zero.

Now consider $i,j\in [k]$ with $i < j$, and we observe that between $f_i$ and $g_i$, in counterclockwise sense, there is neither the arc $f_j$ nor $g_j$.
This holds because, when \Cref{alg:good_v_wiring_simple} adds the pair $(f_i, g_i)$, then the arc $g_i$ appears right after $f_i$ in the circular doubly linked list $L$, i.e., $g_i=\succelem_L(f_i)$.
As neither $f_j$ nor $g_j$ have been deleted at this point because $i<j$, they do not lie between $f_i$ and $g_i$ in counterclockwise sense.
This shows that \Cref{alg:good_v_wiring_simple} returns a good $v$-wiring as claimed.

\smallskip

It remains to show that it can be sped up to run in $O(\deg(v))$ time.
To this end, we show how, after an initial $O(\deg(v))$-time preprocessing step, one can find in each iteration a pair $(a,b)\in \delta^-(v) \times \delta^+(v)$ of arcs in $L$ with $b=\succelem_L(a)$ to add to the $v$-wiring in constant time.
To this end, we maintain throughout the iterations of the while loop of \Cref{alg:good_v_wiring_simple} the set $Z \subseteq \delta^-(v)\cap L$ of all arcs $a\in L\cap \delta^-(v)$ with $\succelem_L(a)\in \delta^+(v)$.
The set $Z$ thus contains the arcs $a$ of all pairs $(a,b)$ that could be chosen in line~\ref{algline:find_ab_pair_simple}.
Before the start of the while loop in \Cref{alg:good_v_wiring_simple}, we initialize $Z$ in $O(\deg(v))$ time by exhaustively going through all arcs in $\delta(v)$. 
Then, at the end of each iteration of the while loop, we update $Z$ in constant time as follows.
If the arc $a$ is deleted in line~\ref{algline:delete_a_simple} of \Cref{alg:good_v_wiring_simple}, we remove it from $Z$ and check whether its predecessor in $L$ should be added to $Z$.
Otherwise, $\succelem_L(a)$ will get deleted in line~\ref{algline:delete_b_simple}, and we check whether, after deleting $\succelem_L(a)$, the arc $a$ still fulfills that its successor in $L$ is a $v$-outgoing arc, i.e., is in $\delta^+(v)$.
If so, we keep $a$ in $Z$; otherwise, we remove it from $Z$.
One can easily verify that this procedure correctly maintains $Z$.
Hence, in each iteration of the while loop, we can find a pair $(a,b)\in Z$ in constant time, and the total time spent on initializing and updating $Z$ is $O(\deg(v))$, as desired.
(These time bounds are for example achieved when implementing $Z$ as a stack.)
\end{proof}

Starting with good $v$-wirings, we now show how to compute an arc-split of $G$ together with a nice $s$-path decomposition as claimed in \Cref{thm:compute_arc_split}.

\begin{proof}[Proof of \Cref{thm:compute_arc_split}]
We start by obtaining a planar combinatorial embedding $\pi^G$ of $G$, which can be done in linear time, and compute a structured $s$-path decomposition of $x$ as described in \Cref{alg:nice_paths_through_good_wirings}.
Note that \Cref{alg:nice_paths_through_good_wirings} requires a planar combinatorial embedding $\pi^G$ of $G$ as it first computes a good $v$-wiring $(W_v,\mu_v)$ for each vertex $v\in V\setminus (T\cup \{s\})$.
Moreover, also line~\ref{algline:chooseArcAtS} relies on $\pi^G_s$.
To perform line~\ref{algline:chooseArcAtS}, we interpret $\pi_s^G$ as a linear order on $\delta^+(s)$ instead of a cyclic one by simply choosing an arbitrary arc $a\in \delta^+(s)$, declaring $a$ to be the first arc, and then numbering the remaining arcs with respect to $\pi^G_s$, i.e.,  in counterclockwise order.

\Cref{alg:nice_paths_through_good_wirings} greedily performs a path-decomposition of the flow $x$.
The algorithm builds paths step-by-step according to these $v$-wirings; whenever we enter a vertex $v$ through some arc $f$, we look for the first pair $(f,g)$ in the $v$-wiring that still has a strictly positive $\mu_v$-value and use the arc $g$ to leave $v$.
To simplify notation, we drop the index $v$ from $\mu_v$ in \Cref{alg:nice_paths_through_good_wirings}.
This is without risk of ambiguity, because each pair $(f,g)$ of consecutive links that we consider is contained in precisely one $W_v$, namely the one for $v=\head(f)=\tail(g)$; hence, in this case $\mu((f,g))\coloneqq \mu_v((f,g))$.

\LinesNumbered
\begin{algorithm2e}[ht]
\DontPrintSemicolon
\caption{Compute structured flow decomposition of $x$.}
\label{alg:nice_paths_through_good_wirings}
Compute a good $v$-wiring $(W_v,\mu_v)$ for each $v\in V\setminus (T\cup \{s\})$.\label{algline:good_v_wirings}\;
$y=x$.\;
$\ell=0$.\;

\While{$y(\delta^+(s))>0$}{
$\ell = \ell+1$.\;
Let $a\in \delta^+(s)$ be the first arc with respect to $\pi^G_s$ in $\delta^+(s)$ with $y(a)>0$.\label{algline:chooseArcAtS}\;
$Q_{\ell}=\{a\}$.\;
\While{$\head(a)\not\in T$}{
Let $(f,g)$ be the first pair in $W_v$ with $f=a$ and $\mu((f,g))>0$, where $v = \head(a)$.;
$Q_\ell = Q_\ell \cup \{g\}$.\;
$a=g$.
}
Let $\lambda_\ell$ be the smallest value $\mu((f,g))$ for any pair $(f,g)$ of consecutive arcs in $Q_\ell$.\label{algline:set_lambda_ell}\;
For every pair $(f,g)$ of consecutive arcs in $Q_\ell$, set $\mu((f,g))=\mu((f,g))-\lambda_\ell$.\label{algline:update_mu}\;
$y(a) = y(a) - \lambda_\ell \qquad \forall a\in Q_{\ell}$.\;
}

\Return $(Q_i,\lambda_i)_{i\in [\ell]}$.\;
\end{algorithm2e}

Given the family of pairs $(Q_i,\lambda_i)_{i\in [\ell]}$ obtained from \Cref{alg:nice_paths_through_good_wirings}, we then construct the claimed arc-split $(H,\pi,\phi)$ of $G$ together with a nice $s$-path decomposition as follows.
Each tuple $(Q_i, \lambda_i)$, for $i\in [\ell]$, will correspond to one path $P_i$ in our $s$-path decomposition with coefficient being $\lambda_i$, i.e., the nice $s$-path decomposition we construct will be $(P_i, \lambda_i)_{i\in [\ell]}$.
We start with an empty graph $H$ on the vertices $V$ and then add, one-by-one for $i$ starting from $1$ to $\ell$, the paths $P_i$.
When considering an index $i\in[\ell]$, we add, for each arc $a=(u,v)\in Q_i$, a fresh arc $h$ from $u$ to $v$ to the graph $H$ and set $\phi(h) = a$.
The path $P_i$ consists of all the arcs added to $H$ when considering the index $i$.
Hence, $P_i$ is a path traversing the same vertices as $Q_i$ (in the same order), simply on a fresh set of arcs.
To conclude the construction, we have to specify how to define the planar combinatorial embedding $\pi$ of $H$.
The idea is that we largely inherit the planar combinatorial embedding $\pi^G$ of $G$.
More precisely, the copy of the arc $(u,v)\in Q_i$ that we added to $H$ will appear in $\pi_u$ and $\pi_v$ at the same spot as it appears in $\pi^G_u$ and $\pi^G_v$, respectively.
In other words, we define the planar combinatorial embedding such that, for any vertex $u\in V$, any subset of non-parallel arcs in $H$ incident to $u$ appear in $\pi_u$ in the precise same order as the order they have in $\pi^G_u$.
It remains to decide how to order different copies in $H$ of the same arc $(u,v)$ of $G$.
Hence, assume that $h_1, h_2$ are two arcs in $H$, both going from $u\in V$ to $v\in V$, and assume that we added $h_1$ in an earlier iteration than $h_2$.
We apply the following rule:
\begin{itemize}
\item Within all arcs from $u$ to $v$ in $H$, the arc $h_1$ comes before $h_2$ in $\pi_u$, i.e., the arc $h_2$ is on the counterclockwise side of $h_1$ within all arcs from $u$ to $v$ in $H$,
\item Within all arcs from $u$ to $v$ in $H$, the arc $h_1$ comes after $h_2$ in $\pi_v$, i.e., the arc $h_2$ is on the clockwise side of $h_1$ within all arcs from $u$ to $v$ in $H$.
\end{itemize}

This finishes the procedure to compute the arc-split $(H,\pi,\phi)$ and the $s$-path decomposition $(P_i,\lambda_i)_{i\in [\ell]}$.

We start by showing the desired bound of $\ell = O(|V|)$ on the number of paths in our $s$-path decomposition.
Note that this number of paths is the same as the number of paths $Q_i$, computed in \Cref{alg:nice_paths_through_good_wirings}.
Each time we compute a path $Q_i$ in \Cref{alg:nice_paths_through_good_wirings}, one of the values $\mu_v((f,g))$ is set to zero for some pair $(f,g)$ in one of the sets $W_v$ for some $v\in V\setminus (T\cup \{s\})$.
We thus have
\begin{equation*}
\ell \leq \sum_{v\in V\setminus (T\cup \{s\})} |W_v| = O\left(\sum_{v\in V \setminus (T\cup \{s\})} \deg(v)\right) = O(|A|) = O(|V|),
\end{equation*}
where the first equality follows from $|W_v|\leq \deg(v)$, which holds by \Cref{lem:comp_good_v_wiring}, the second equality is due to the fact that the sum of all degrees of a graph is twice the number of its arcs, and the last inequality holds because $G$ is planar. 

We now discuss the running time of the suggested algorithm.
First, note that \Cref{alg:nice_paths_through_good_wirings} can clearly be implemented to run in $O(|V|^2)$ time.
More precisely, due to \Cref{lem:comp_good_v_wiring}, computing the good $v$-wirings at the beginning of the algorithm takes $O(\sum_{v\in V\setminus (T\cup \{s\})} \deg(v)) = O(|A|) = O(|V|)$ time.
Each iteration of the outer while loop constructs one path $Q_i$, and as we construct $\ell = O(|V|)$ many paths, this while loop has $O(|V|)$ iterations.
Moreover, the inner while loop successively adds arcs to the path $Q_i$, each such step taking $O(1)$ time.
Because $|Q_i| = O(|V|)$ for each path $Q_i$, this leads to a total running time of $O(|V|^2)$ for \Cref{alg:nice_paths_through_good_wirings}.
Finally, the construction of the arc-split $(H,\pi,\phi)$ and nice $s$-path decomposition $(P_i,\lambda_i)_{i\in [\ell]}$ takes time linear in the size of $H$, which is bounded by $O(\sum_{i=1}^{\ell} |Q_i|) = O(|V|^2)$ as desired, where the equality follows from $\ell = O(|V|)$ and $|Q_i|\leq |V|$ for $i\in [\ell]$.

It remains to show that the arc-split $(H,\pi,\phi)$ and the $s$-path decomposition $(P_i,\lambda_i)_{i\in [\ell]}$ have the desired properties.
First, note that $\pi$ is indeed a planar combinatorial embedding of $H$.
This readily follows from the fact that a (geometric) planar embedding of $H$ corresponding to $\pi$ can be obtained by first obtaining a (geometric) planar embedding of $G$ that corresponds to $\pi^G$ and then replacing arcs of $G$ by parallel arcs.
Also, the properties of $\phi$ for $(H,\pi,\phi)$ to be an arc-split of $G$ (see \Cref{def:arc_split_graph}) are clearly fulfilled.
Hence, $(H,\pi,\phi)$ is an arc-split of $G$.

Because the arcs of $H$ are by construction the disjoint union of the paths $P_i$ for $i\in [\ell]$, and all these paths start at $s$, we have that $P_1,\dots, P_{\ell}$ is an $s$-path partitioning of $H$.
To show that it is a nice $s$-path partitioning, we have to show that they are pairwise non-crossing.
Hence, let $i,j\in [\ell]$ with $i < j$ and let $v\in V$ be a common internal vertex of $P_i$ and $P_j$.
We name the arcs of $P_i$ and $P_j$ that are incident to $v$ as follows:
\begin{align*}
\{a_i^H\} &= P_i \cap \delta^-(v) & a_i^G &= \phi(a_i^H)\\
\{b_i^H\} &= P_i \cap \delta^+(v) & b_i^G &= \phi(b_i^H)\\
\{a_j^H\} &= P_j \cap \delta^-(v) & a_j^G &= \phi(a_j^H)\\
\{b_j^H\} &= P_j \cap \delta^+(v) & b_j^G &= \phi(b_j^H).
\end{align*}
Note that we may have $a_i^G = a_j^G$, or $b_i^G = b_j^G$, or both.
We need to show that neither $(a_i^H, a_j^H, b_i^H, b_j^H)$ nor $(a_j^H, a_i^H, b_j^H, b_i^H)$ is a $\pi_v$-progression.
To this end, we exploit that, by construction of $\pi_v$, we have the following property:
\begin{equation}\label{eq:orderCompatibilityGH}
\parbox[c]{0.8\linewidth}{For any subset $U\subseteq \delta_H(v)$ of non-parallel arcs incident to $v$, the ordering of $U$\\
with respect to $\pi_v$ is the same as the ordering of $\phi(U)$ with respect to $\pi_v^G$.}
\end{equation}

We first observe that $(a_i^H, a_j^H, b_i^H, b_j^H)$ is not a $\pi_v$-progression.
If $a_i^G \neq a_j^G$, then $(a_i^G, a_j^G, b_i^G)$ is not a $\pi_v^G$-progression, which follows from $(W_v,\mu_v)$ being a good $v$-wiring.
By~\eqref{eq:orderCompatibilityGH}, this implies that $(a_i^H, a_j^H, b_i^H)$ is not a $\pi_v$-progression, and therefore neither is $(a_i^H, a_j^H, b_i^H, b_j^H)$.
Otherwise, if $a_i^G = a_j^G$, then, in the ordering $\pi_v$ of the arcs $\phi^{-1}(a_i^G)$, the arc $a_i^H$ appears after $a_j^H$.
Again, this implies that $(a_i^H, a_j^H, b_i^H, b_j^H)$ is not a $\pi_v$-progression as $b_i^H, b_j^H \not\in \phi^{-1}(a_i^G)$.

Analogously, we now show that $(a_j^H, a_i^H, b_j^H, b_i^H)$ is not a $\pi_v$-progression.
If $b_i^G \neq b_j^G$, then $(a_i^G, b_j^G, b_i^G)$ is not a $\pi_v^G$-progression, which follows from $(W_v,\mu_v)$ being a good $v$-wiring.
By~\eqref{eq:orderCompatibilityGH}, this implies that $(a_i^H, b_j^H, b_i^H)$ is not a $\pi_v$-progression, and therefore neither is $(a_j^H, a_i^H, b_j^H, b_i^H)$.
Otherwise, if $b_i^G = b_j^G$, then, in the ordering $\pi_v$ of the arcs $\phi^{-1}(b_i^G)$, the arc $b_i^H$ appears before $b_j^H$.
Again, this implies that $(a_j^H, a_i^H, b_j^H, b_i^H)$ is not a $\pi_v$-progression.

Thus, $P_1,\ldots, P_\ell$ are pairwise non-crossing, which shows that $P_1,\dots, P_\ell$ is a nice $s$-path partitioning.
To finish the proof, it remains to show that
\begin{equation}\label{eq:PisArePathDecomp}
\sum_{i=1}^{\ell} \lambda_i \chi^{\phi(P_i)} = x.
\end{equation}

This holds because \Cref{alg:nice_paths_through_good_wirings} performs greedily an $s$-path decomposition that uses up all the $\mu_v$-values of the good $v$-wirings $(W_v, \mu_v)$ for each $v\in V\setminus (T\cup \{s\})$.
Indeed, this path decomposition is just one way to perform a greedy path decomposition of the original flow $x$.
(Recall that a good $v$-wiring $(W_v, \mu_v)$ with $W_v=((f_1,g_1),\dots, (f_k,g_k))$ satisfies
$\sum_{\substack{i\in [k]: a\in \{f_i,g_i\}}} \mu_v((f_i,g_i)) = x(a)$ for all $a\in \delta(v)$, and thus the $\mu_v$-values indeed represent the original flow.)
Moreover, a greedy path decomposition of a flow in an acyclic graph always decomposes the full flow value on all arcs.
This implies \Cref{eq:PisArePathDecomp}.
\end{proof}
\section{Nice path partitionings have nice properties}\label{sec:properties_path_decomp}

In this section we prove \Cref{lem:non-interleaving} and \Cref{lem:discr_equiv_to_interval}, which are both statements about source-numbered nice $s$-path partitionings stemming from a PSSUF instance $(G,s,T,d,x)$.
As these statements do not depend on the fractional flow $x$ or the demands $d$, we will only fix, throughout this section, an acyclic graph $G=(V,A)$ and an arc-split $(H=(V,F), \pi , \phi)$ of $G$ together with a source-numbered nice $s$-path partitioning $P_1, \ldots , P_{\ell}$ of $H$.
Moreover, we assume throughout that no end vertex of any path of the $s$-path partitioning has outgoing arcs, because these end vertices correspond to terminals, which do not have outgoing arcs by assumption.
We thus also call here such end vertices \emph{terminals} and let $T$ be the set of all such terminals.

Our proofs of the statements \Cref{lem:non-interleaving} and \Cref{lem:discr_equiv_to_interval} both rely on a crucial property how the paths $P_1,\dots, P_{\ell}$ can interact.
This property is described in terms of how two paths that share a common internal vertex split the planar embedding in two parts.
We represent these two parts as a tri-coloring, which we introduce first.
Using this tri-coloring, we first prove \Cref{lem:discr_equiv_to_interval} in \Cref{sec:discr_equiv}, and then \Cref{lem:non-interleaving} in \Cref{sec:non-interleaving}.

\paragraph{Tri-coloring}

Let us fix a planar geometric embedding of $H$ realizing the combinatorial embedding $\pi$. 
Let $\widetilde{P}$ and $\overline{P}$ be two distinct paths of the nice $s$-path partitioning that share a vertex $v \in V \setminus \{s\}$.
Moreover, let $C$ be the subgraph of $H$ obtained from taking the union of the $s$-$v$ subpaths of $\widetilde{P}$ and $\overline{P}$.
Observe that $C$ is a connected graph in which every vertex $v$ has even degree $|\delta(v)|$.
Thus, the planar dual $C^*$ of the undirected graph obtained from $C$ by ignoring the orientation of the arcs is bipartite.\footnote{This is well known and can be seen as follows.
In the undirected graph obtained from $C$ by ignoring the orientation of the arcs, all vertex degrees are even and hence every cut contains an even number of edges. 
As cycles in the dual graph $C^*$ correspond to cuts in this graph, $C^*$ has no odd cycle and is thus bipartite.}
Therefore, we can color the faces of $C$ with two colors, red and blue.
Because the geometric embedding is planar, this yields a red-blue coloring of all arcs of $H$ except for $C$.
Finally, we color the arcs in $C$ in black.
We call this coloring of all arcs of $H$ the \emph{tri-coloring of $H$ induced by $\widetilde{P}$, $\overline{P}$, and $v$}.
(It is unique up to exchanging red and blue.)
See \Cref{fig:coloring_example}.
When we say that a path changes color, we mean that it changes the color of its arcs from either red, blue, or black, to another one of these 3~options.

\begin{figure}[ht]
\begin{center}
\begin{tikzpicture}[
vertices/.style={circle,draw=black,inner sep=2pt, fill=white},
arcs/.style={line width=2pt,-stealth},
p1/.style={},
p2/.style={dashed}
]

\begin{scope}
\coordinate (c) at (-1,0);
\coordinate (a1) at (-2,0.5);
\coordinate (a2) at (-1,1);
\coordinate (a3) at (-2,2.5);
\coordinate (a4) at (-2,6);
\coordinate (a5) at (1,6);
\coordinate (a6) at (2,4);
\coordinate (a7) at (2,2);
\coordinate (a8) at (1.2,2.5);
\coordinate (a9) at (0.9,4.5);
\coordinate (a10) at (0,4);

\coordinate (b1) at (a2);
\coordinate (b2) at (3,1);
\coordinate (b3) at (5,3);
\coordinate (b4) at (a6);
\coordinate (b5) at (3,2.5);
\coordinate (b6) at (a7);
\coordinate (b7) at (a3);
\coordinate (b8) at (-0.9,4.5);
\coordinate (b9) at (0,3);
\coordinate (b10) at (a8);
\coordinate (b11) at (a10);
\end{scope}

\draw[fill=blue, opacity=0.3]  (c) -- (a1) -- (a2)  --cycle;
\draw[fill=blue, opacity=0.3]  (a2) -- (a3) -- (a7) -- (b5) -- (a6) -- (b3) -- (b2)  --cycle;
\draw[fill=blue, opacity=0.3]  (a3) -- (a4) -- (a5) -- (a6) -- (a7) -- (a8) -- (a9) -- (a10) -- (a8) -- (b9) -- (b8)  --cycle;

\draw[fill=red, opacity=0.3] (a7) -- (a6) --(b5) -- cycle;
\draw[fill=red, opacity=0.3]  (b6) -- (b7) -- (b8) -- (b9) -- (b10) -- (b11) -- (a9) -- (a8) --cycle;
\draw[draw=none, fill=red, opacity=0.3]  (c) -- (-3,0) -- (-3,6.5) -- (5.5,6.5) -- (5.5,0) -- (c) -- (b1) -- (b2) -- (b3) --( b4) -- (a5) -- (a4) -- (a3) -- (a2) -- (a1) -- cycle;
\draw[draw=none, fill=red, opacity=0.3]  (-3,0) -- (5.5,0) -- (5.5,-0.5) -- (-3,-0.5) --cycle;

\begin{scope}[every node/.style={vertices}]
\node  (s) at (c) {$s$};
\foreach \i in {1, ..., 9} {
\node (n\i) at (a\i) {};
}
\foreach \i in {1, ..., 10} {
\node (m\i) at (b\i) {};
}
\node (m11) at (b11) {$v$};
\end{scope}
\begin{scope}[arcs, p1]
\draw (s) to (n1);
\draw (n1) to (n2);
\draw (n2) to (n3);
\draw (n3) to (n4);
\draw (n4) to (n5);
\draw (n5) to (n6);
\draw (n6) to (n7);
\draw (n7) to (n8);
\draw (n8) to (n9);
\draw (n9) to (m11);
\end{scope}
\begin{scope}[arcs, p2]
\draw (s) to (m1);
\draw (m1) to (m2);
\draw (m2) to (m3);
\draw (m3) to (m4);
\draw (m4) to (m5);
\draw (m5) to (m6);
\draw (m6) to (m7);
\draw (m7) to (m8);
\draw (m8) to (m9);
\draw (m9) to (m10);
\draw (m10) to (m11);
\end{scope}

\end{tikzpicture} \end{center}
\caption{Example of the coloring of the plane. The $s$-$v$ subpaths of $\overline{P}$ and $\widetilde{P}$ are shown in black with solid and dashed arcs, respectively.
The color assigned to the arcs in $\overline{P}\cup \widetilde{P}$ is black, arcs lying within a red face are colored red, and arcs lying within a blue face are colored blue.
}
    \label{fig:coloring_example}
\end{figure}

Note that no path of the nice $s$-path partitioning except for $\widetilde{P}$ and $\overline{P}$ contains black arcs, because paths in $P_1,\dots, P_{\ell}$ use distinct arcs.
The paths $\widetilde{P}$ and $\overline{P}$ use black arcs until they reach $v$, and red or blue ones afterward.

The following lemma highlights a crucial property how different paths of a nice $s$-partitioning can interact, based on the above-introduced tri-coloring.
\begin{lemma}
\label{lem:only_one_change_of_color}
Let $\widetilde{P}$ and $\overline{P}$ be two distinct paths of the nice $s$-path partitioning that share a vertex $v \in V \setminus \{s\}$ and consider the tri-coloring of $H$ induced by $\widetilde{P}$, $\overline{P}$, and $v$.
Then any path $P$ of the nice $s$-path partitioning can change its color at most once and, if so, only at vertex $v$.
\end{lemma}

\begin{proof}
As in the definition of a tri-coloring, let $C$ be the set of all arcs colored black, i.e., let $C$ be the union of the $s$-$v$ subpaths of $\overline{P}$ and $\widetilde{P}$.
First assume that $P \not\in \{\overline{P},\widetilde{P}\}$.
Suppose for the sake of deriving a contradiction that the color of $P$ changes at a vertex $u\neq v$. 
Note that $P$ does not contain any black arcs since it is arc-disjoint from $\overline{P}$ and $\widetilde{P}$.
Because $P$ starts at $s$ we have $u \neq s$.
Moreover, for a color change to be possible at $u$, the vertex $u$ must lie on  the $s$-$v$ subpath of either $\overline{P}$ or $\widetilde{P}$.
Thus, we either have $|C \cap \delta(u)|=2$, in which case $u$ is visited by exactly one of the $s$-$v$ subpaths of $\overline{P}$ and $\widetilde{P}$, or we have $|C \cap \delta(u)|=4$, in which case $u$ is visited both by the $s$-$v$ subpath of $\overline{P}$ and the $s$-$v$ subpath of $\widetilde{P}$.
This holds because any path, including $\overline{P}$ and $\widetilde{P}$, can visit any vertex at most once.
If $|C \cap \delta(u)|=2$, then when changing colors at $u$ the path $P$ would cross the path $Q \in \{\overline{P},\widetilde{P}\}$ whose $s$-$v$ subpath visits $u$. More formally, let $a$ be the single arc in $Q\cap \delta^-(u)$ and $b$ be the single arc in $Q\cap \delta^+(v)$.
The arcs counterclockwise in $\pi_u$ between $a$ and $b$ must either all be red or blue (say red), and all arcs clockwise between $a$ and $b$ have the other color among red/blue (say blue), see the left part of \Cref{fig:changing_color_only_at_v}.
Up to symmetry this means $\pi_u$ has the following structure:
\[
\pi_u = (a, \underbrace{\ldots}_{\textcolor{red}{red}} , b , \underbrace{\ldots}_{\textcolor{blue}{blue}}).
\]
Thus, a path entering $u$ with an arc of color red or blue and leaving $u$ with an arc of the other of these colors implies a crossing among the paths $P_1,\dots, P_\ell$, which violates that they are a nice $s$-path partitioning.

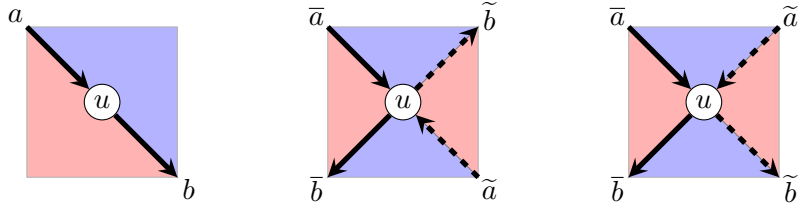
\begin{figure}[ht]
\begin{center}
\begin{tikzpicture}[
vertices/.style={circle,draw=black,inner sep=2pt, fill=white},
arcs/.style={line width=2pt,-stealth},
p1/.style={},
p2/.style={dashed}
]

\begin{scope}[shift={(-4,0)}]
\begin{scope}
\coordinate (c) at (0,0);
\coordinate (u1) at (-1, 1);
\coordinate (u2) at (-1,-1);
\coordinate (w1) at ( 1,-1);
\coordinate (w2) at ( 1, 1);
\end{scope}
\draw[fill=blue, opacity=0.3]  (c) -- (u1) -- (w2) -- (w1) --cycle;
\draw[fill=red, opacity=0.3]  (c) -- (w1) -- (u2) -- (u1) --cycle;
\begin{scope}[every node/.style={vertices}]
\node (u) at (c) {$u$};
\end{scope}
\begin{scope}[arcs, p1]
\draw (u1) --  (u);
\draw (u) --  (w1);
\end{scope}
\node at (-1.15,1.15) {$a$};
\node at (1.15,-1.15) {$b$};
\end{scope}

\begin{scope}[shift={(4,0)}]
\begin{scope}
\coordinate (u1) at (-1, 1);
\coordinate (u2) at (-1,-1);
\coordinate (w1) at ( 1,-1);
\coordinate (w2) at ( 1, 1);
\end{scope}
\draw[fill=blue, opacity=0.3]  (0,0) -- (u1) -- (w2) --cycle;
\draw[fill=blue, opacity=0.3]  (0,0) -- (u2) -- (w1) --cycle;
\draw[fill=red, opacity=0.3]  (0,0) -- (w1) -- (w2) --cycle;
\draw[fill=red, opacity=0.3]  (0,0) -- (u1) -- (u2) --cycle;
\begin{scope}[every node/.style={vertices}]
\node (u) at (0,0) {$u$};
\end{scope}
\begin{scope}[arcs]
\begin{scope}[p1]
\draw (u1) --  (u);
\draw (u) --  (u2);
\end{scope}
\begin{scope}[p2]
\draw (w2) --  (u);
\draw (u) --  (w1);
\end{scope}
\node at (-1.15,1.15) {$\overline{a}$};
\node at (-1.15,-1.15) {$\overline{b}$};
\node at (1.15,-1.15) {$\widetilde{b}$};
\node at (1.15,1.15) {$\widetilde{a}$};
\end{scope}
\end{scope}

\begin{scope}[shift={(0,0)}]
\begin{scope}
\coordinate (c) at (0,0);
\coordinate (u1) at (-1, 1);
\coordinate (u2) at (-1,-1);
\coordinate (w1) at ( 1,-1);
\coordinate (w2) at ( 1, 1);
\end{scope}
\draw[fill=blue, opacity=0.3]  (c) -- (u1) -- (w2) --cycle;
\draw[fill=blue, opacity=0.3]  (c) -- (u2) -- (w1) --cycle;
\draw[fill=red, opacity=0.3]  (c) -- (w1) -- (w2) --cycle;
\draw[fill=red, opacity=0.3]  (c) -- (u1) -- (u2) --cycle;
\begin{scope}[every node/.style={vertices}]
\node (u) at (c) {$u$};
\end{scope}
\begin{scope}[arcs]
\begin{scope}[p1]
\draw (u1) --  (u);
\draw (u) --  (u2);
\end{scope}
\begin{scope}[p2]
\draw (w1) --  (u);
\draw (u) --  (w2);
\end{scope}
\node at (-1.15,1.15) {$\overline{a}$};
\node at (-1.15,-1.15) {$\overline{b}$};
\node at (1.15,-1.15) {$\widetilde{a}$};
\node at (1.15,1.15) {$\widetilde{b}$};
\end{scope}
\end{scope}

\end{tikzpicture} \end{center}
\caption{Illustration of the neighborhood of a vertex $u$ in $C$.
The three different possible situations (up to symmetry and exchange of the colors red and blue) are illustrated in the picture. 
Here the solid and dashed arcs belong to the $s$-$v$ subpaths of $\overline{P}$ and $\widetilde{P}$, respectively.
A path changing its color from a red face to a blue face (or the other way around) at $u$ crosses one of the paths $\overline P$ and $\widetilde P$.}
    \label{fig:changing_color_only_at_v}
\end{figure}

So assume $|C \cap \delta(u)|=4$.
Hence, the four primal faces adjacent to $u$ have alternating colors between red and blue.
(More precisely, as we color the arcs, the color of a face corresponds to the color of all arcs within that face.)
Let us denote by $\overline{a}$ and $\widetilde{a}$ the arc of $\overline{P}$ and $\widetilde{P}$, respectively, that enters $u$, and by $\overline{b}$ and $\widetilde{b}$ the arc of $\overline{P}$ and $\widetilde{P}$, respectively, that leaves $u$, i.e., $\{\overline{a}\} = \overline{P}\cap \delta^-(u)$, $\{\widetilde{a}\} = \widetilde{P}\cap \delta^-(u)$, $\{\overline{b}\} = \overline{P}\cap \delta^+(u)$, and $\{\widetilde{b}\} = \widetilde{P}\cap \delta^+(u)$.
Because $\overline{P}$ and $\widetilde{P}$ do not cross, we can describe the structure of $\pi_u$ up to symmetries as follows:
\begin{align*}
\pi_u &= (\overline{a}, \underbrace{\ldots}_{\textcolor{red}{red}} , \overline{b}, \underbrace{\ldots}_{\textcolor{blue}{blue}}, \widetilde{a}, \underbrace{\ldots}_{\textcolor{red}{red}} , \widetilde{b}, \underbrace{\ldots}_{\textcolor{blue}{blue}}) \quad \text{ or }\\
\pi_u &= (\overline{a}, \underbrace{\ldots}_{\textcolor{red}{red}} , \overline{b}, \underbrace{\ldots}_{\textcolor{blue}{blue}}, \widetilde{b}, \underbrace{\ldots}_{\textcolor{red}{red}} , \widetilde{a}, \underbrace{\ldots}_{\textcolor{blue}{blue}}).
\end{align*}
Again, a path $P$ would have to cross one of the paths $\overline{P}$ and $\widetilde{P}$ when changing color in $u$.
See \Cref{fig:changing_color_only_at_v} for a visualization. This is not possible in a nice $s$-path partitioning.

We conclude that $P$ can change its color only at the vertex $v$.
Because $P$ is a path, it can thus change its color at most once.

Now suppose $P\in \{\overline{P},\widetilde{P}\}$, and let $Q \in \{\overline{P},\widetilde{P}\}$ be the other path.
By definition, $P$ is black until $v$ and becomes red or blue after that. 
A second change of color is thus only possible at a vertex that $P$ visits after $v$.
Because $P$ cannot visit a vertex twice, such a vertex is not visited by the $s$-$v$ subpath of $P$.
Thus, if $P$ changes its color at a vertex $u\neq v$, the vertex $u$ must be an internal vertex of the $s$-$v$ subpath of $Q$, and we have $|C \cap \delta(u)|=2$.
However, at such a vertex $u$, the path $P$ cannot change color as this would imply a crossing of $P$ and $Q$ at $u$.
Thus, after changing color from black to red or blue at vertex $v$, the path $P$ does not change color again, as desired.
\end{proof}

\subsection{Characterization of paths sharing an arc}\label{sec:discr_equiv} 

We now derive \Cref{lem:discr_equiv_to_interval} from \Cref{lem:only_one_change_of_color}.
For convenience, we restate the lemma below and recall that, for $t\in T$, the set
\begin{equation*}
\mathcal{P}^t \coloneqq \{P_i\colon i\in [\ell] \text{ and $P_i$ ends at terminal $t$}\},
\end{equation*}
are all paths in the nice $s$-path partitioning that end at terminal $t$.

\lemDiscrEquivToInterval*

\begin{proof}
Fix some arc $a=(v,w) \in A$ and a geometric planar embedding of $H$ realizing the combinatorial embedding $\pi$.
Let
\begin{align*}
I           &\coloneqq \{i\in [\ell] \colon a\in \phi(P_i)\} \text{ and}\\
\mathcal{P} &\coloneqq \{P_i \colon i\in I \}
\end{align*}
be the indices that we want to show to be discrepancy-equivalent to a circular interval, and the corresponding paths, respectively.

If $v=s$, then $I$ is a circular interval because the paths $P_1, \dots, P_{\ell}$ are source-numbered and the arcs in $\phi^{-1}(a)$ appear consecutively in $\pi_s$ (see \Cref{def:arc_split_graph}, i.e., the definition of an arc-split graph).
Moreover, if $|\mathcal{P}| \le 1$ the statement of the lemma is trivial.
Hence, we will now assume $s \neq v$ and $|\mathcal{P}| \ge 2$.

By definition of $\mathcal{P}$, every path in $\mathcal{P}$ uses an arc of $\phi^{-1}(a) \subseteq \delta^+(v)$. Since $v \neq s$, also every path of $\mathcal{P}$ uses an arc of $\delta^-(v)$. We denote these arcs as $\Fin$, i.e., $\Fin \subseteq F$ contains all arcs that belong to a path of $\mathcal{P}$ and enter $v$.

Recall that by the definition of an arc-split graph, the arcs in $\phi^{-1}(a)$ appear consecutively in $\pi_v$.
Consider the first arc of $\Fin$ appearing after the arcs of $\phi^{-1}(a)$ in clockwise direction in the geometric embedding and let $\Pclo$ be the unique path in  $\mathcal{P}$ using this arc. (Uniqueness follows from the fact that the different paths in $\mathcal{P}$ are arc-disjoint by definition of an $s$-path partitioning.) 
Similarly, consider the first arc of $\Fin$ after the arcs of $\phi^{-1}(a)$ in counterclockwise direction and let $\Pcc$ be the unique path using this arc. 
Because $|\mathcal{P}| \ge 2$, also $\abs{\Fin} \geq 2$ and, thus, the paths $\Pcc$ and $\Pclo$ exist and are distinct. 

Both $\Pclo$ and $\Pcc$ start in $s$ and contain the vertex $v$. We consider the tri-coloring of $H$ induced by $\Pclo$, $\Pcc$, and $v$.
Note that, for each terminal $t\in T$, either all arcs in $\delta^-(t)$ are red or all arcs in $\delta^-(t)$ are blue.
This follows from the fact that terminals cannot lie in the interior of the paths $\Pclo$ or $\Pcc$, because they do not have outgoing arcs. 
In particular, $v$ cannot be a terminal, because it has an outgoing arc $a=(v,w)$.
We call a terminal $t\in T$ \emph{red} if all arcs in $\delta^-(t)$ are red, and we call it \emph{blue} if all arcs in $\delta^-(t)$ are blue.
Analogously, the vertex $w\coloneqq \head(a)$ does not lie on the $s$-$v$ subpaths of either $\Pclo$ or $\Pcc$ because both paths $\Pclo$ and $\Pcc$ contain some arc from $\phi^{-1}((v,w))$, and $w$ cannot appear twice on any path.
Hence, also the arcs in $\delta(w)$ are either all red or all blue.
We assume without loss of generality that the arcs in $\delta(w)$ are all colored red. 
See \Cref{fig:v} for an illustration.

\begin{figure}[ht]
\begin{center}
\begin{tikzpicture}[ scale =1.5,
vertices/.style={circle,draw=black,inner sep=2pt, fill=white},
arcs/.style={line width=2pt,-stealth},
p1/.style={},
p2/.style={dashed}
]

\begin{scope}
\coordinate (u1) at (-1, 1);
\coordinate (u2) at (-1,-1);
\coordinate (w1) at ( 1,-1);
\coordinate (w2) at ( 1, 1);
\end{scope}
\draw[fill=blue, opacity=0.3]  (0,0) -- (u1) -- (w2) --cycle;
\draw[fill=red, opacity=0.3]  (0,0) -- (u1) -- (u2) -- (w1) -- (w2) --cycle;
\begin{scope}[every node/.style={vertices}]
\node (v) at (0,0) {$v$};
\node (w) at (0,-0.8) {$w$};
\end{scope}
\begin{scope}[arcs]
\begin{scope}[p1]
\draw (u1) --  (v);
\draw[->, red!70!black] (v) to [bend right=45]  (w);
\end{scope}
\begin{scope}[p2]
\draw (w2) --  (v);
\draw[->, red!70!black] (v) to [bend left = 45]  (w);
\end{scope}
\end{scope}

\end{tikzpicture} \end{center}
\caption{Illustration of the neighborhood of $v$. 
The solid arcs are arcs of $\Pclo$ and the dashed arcs belong to $\Pcc$.}
    \label{fig:v}
\end{figure}
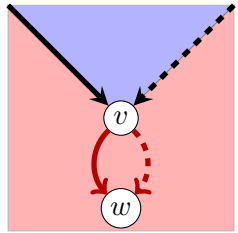

Let  $\mathcal{Q}$ be the set of paths containing $\Pclo$, $\Pcc$, and all paths $P \in \{P_1,\dots, P_{\ell}\}$ whose first arc is colored blue.
By construction of the tri-coloring and since $s$ is adjacent to only two black arcs, %
the index set $J\coloneqq \{i\in [\ell] \colon P_i \in \mathcal{Q}\}$ is a circular interval.
We claim that the set $I$ of indices of paths in $\mathcal{P}$ is discrepancy-equivalent to the set $J$ of indices of paths in $\mathcal{Q}$.
To this end, we will show that:
\begin{enumerate}
\item\label{item:redPaths} For every red terminal $t\in T$, we have $\mathcal{P}^t \cap \mathcal{P} =\mathcal{P}^t \cap \mathcal{Q}$.
\item\label{item:bluePaths} For every blue terminal $t\in T$, we have $\mathcal{P}^t \cap \mathcal{P} = \emptyset$ and $\mathcal{P}^t \subseteq \mathcal{Q}$.
\end{enumerate}
Hence, $J$ can be obtained from $I$ by $S^t$-addition for all terminals $t\in T$ that are colored blue, where $S^t\coloneqq \{ i \in [\ell] : P_i \in \mathcal{P}^t\}$.

It remains to prove \ref{item:redPaths} and \ref{item:bluePaths}.
By \Cref{lem:only_one_change_of_color}, a path $P\in \{P_1,\dots, P_{\ell}\}$ can only change its color at $v$. We show that this implies that $P$ can only change its color from either blue or black to red, but not the other way around.
Hence, once a path is red, it stays red.
Because all arcs incident with $w$ are red, including all arcs in $\phi^{-1}(a)$, this implies \ref{item:bluePaths}.
More concretely, given a blue terminal $t$, a path $P \in \mathcal{P}^t$ can never use a red arc, because it would then stay red and can therefore not end in a blue terminal.
Thus, $P \in \mathcal{Q}$.
Moreover, $P \not\in \mathcal{P}$ since otherwise it contains $w$, which only has red arcs incident with it.
Note that this also implies that both $\Pclo$ and $\Pcc$ do not have a blue terminal since they use an arc of $\phi^{-1}(a)$, which is red.

To show that a path can only change the color from blue or black to red but not the other way around, remember that paths can only change their color in $v$ by \Cref{lem:only_one_change_of_color}.
The vertex $v$ has only two adjacent black arcs (one incoming arc each from $\Pclo$ and $\Pcc$).
Both $\Pclo$ and $\Pcc$ use an arc of $\phi^{-1}(a)$ by definition and all the arcs of $\phi^{-1}(a)$ appear consecutively in the embedding.
Thus, $\pi_v$ has the following form,
where $\aclo$ and $\acc$ are the arcs of $\Pclo$ and $\Pcc$, respectively, that enter $v$, and $\bclo$ and $\bcc$ are the arcs of $\Pclo$ and $\Pcc$, respectively, that leave $v$, i.e., $\{\aclo\} = \Pclo\cap \delta^-(v)$, $\{\bclo\} = \Pclo\cap \delta^+(v)$, $\{\acc\}=\Pcc\cap \delta^-(v)$, and $\{\bcc\}=\Pcc\cap \delta^+(v)$:
\begin{equation}
\label{eq:neighborhood_v}
\pi_v = (\aclo, \underbrace{\ldots}_{\textcolor{red}{red}} ,  \underbrace{\bclo, \ldots, \bcc}_{\textcolor{red}{\subseteq \phi^{-1}(a)}}, \underbrace{\ldots}_{\textcolor{red}{red}} , \acc, \underbrace{\ldots}_{\textcolor{blue}{blue}}).
\end{equation}
One can even show that all arcs of $\phi^{-1}(a)$ lie counterclockwise between $\bclo$ and $\bcc$, but since it is not needed for the proof, we just use the simpler fact that $\bclo, \bcc \in \phi^{-1}(a)$ and thus also all arcs which are counterclockwise between them are contained in $\phi^{-1}(a)$.
See \Cref{fig:v_with_arc_names} for an illustration of the neighborhood of $v$.

\begin{figure}[ht]
\begin{center}
\begin{tikzpicture}[ scale =1.5,
vertices/.style={circle,draw=black,inner sep=2pt, fill=white},
arcs/.style={line width=2pt,-stealth},
p1/.style={},
p2/.style={dashed}
]

\begin{scope}
\coordinate (u1) at (-1, 1);
\coordinate (u2) at (-1,-1);
\coordinate (w1) at ( 1,-1);
\coordinate (w2) at ( 1, 1);
\end{scope}
\draw[fill=blue, opacity=0.3]  (0,0) -- (u1) -- (w2) --cycle;
\draw[fill=red, opacity=0.3]  (0,0) -- (u1) -- (u2) -- (w1) -- (w2) --cycle;
\begin{scope}[every node/.style={vertices}]
\node (v) at (0,0) {$v$};
\node (w) at (0,-0.8) {$w$};
\end{scope}
\begin{scope}[arcs]
\begin{scope}[p1]
\draw (u1) -- node[below left=-7pt]{$\aclo$} (v);
\draw[->, red!70!black] (v) to [bend right=45] node[left]{$\bclo$}  (w);
\end{scope}
\begin{scope}[p2]
\draw (w2) -- node[below right=-5pt]{$\acc$}  (v);
\draw[->, red!70!black] (v) to [bend left = 45] node[right]{$\bcc$}  (w);
\end{scope}
\end{scope}

\end{tikzpicture} \end{center}
\caption{Illustration of the arcs $\aclo, \bclo,\acc,\bcc$ incident to $v$.
The solid arcs are arcs of $\Pclo$ and the dashed arcs belong to $\Pcc$.}
    \label{fig:v_with_arc_names}
\end{figure}
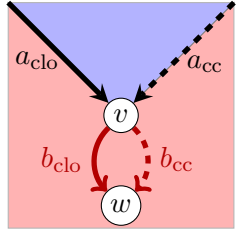

\noindent To change the color from red to blue, a path cannot use an arc of $\phi^{-1}(a)$, since they are all red and leave $v$. This means that any path changing color from red to blue needs to cross either $\Pclo$ or $\Pcc$, which is forbidden by the definition of a nice $s$-path partitioning. 

It remains to prove \ref{item:redPaths}.
To this end we show that every $P\in \{P_1,\dots, P_{\ell}\}$ satisfies the following.
\begin{enumerate}[label=(\alph*)]
\item \label{item:redPathsDoNotvisitV} If $P$ enters $v$ via a red arc, then it does not use an arc from $\phi^{-1}(a)$, i.e., $P\notin \mathcal{P}$.
\item \label{item:colorChangeImpliesUsingA} If $P$ changes its color from blue to red, then it must use an arc from $\phi^{-1}(a)$, i.e., $P\in \mathcal{P}$.
\end{enumerate}
Then \ref{item:redPaths} follows, because \ref{item:redPathsDoNotvisitV} implies that every path in $\mathcal{P}$, i.e., every path using an arc in $\phi^{-1}(a)$ from $v$ to $w$, must change its color at $v$ from blue to red and thus must be contained in $\mathcal{Q}$. 
Hence, $\mathcal{P} \subseteq \mathcal{Q}$, and thus $\mathcal{P}^t \cap \mathcal{P} \subseteq \mathcal{P}^t \cap \mathcal{Q}$ for every terminal $t\in T$.
(Here we used that  $\Pcc$ and $\Pclo$ are contained in $\mathcal{Q}$ by definition.)
Now consider a path $P \in \mathcal{P}^t \cap \mathcal{Q}$ for some red terminal $t\in T$.
Then \ref{item:colorChangeImpliesUsingA} implies that $P$ uses an arc from $\phi^{-1}(a)$, i.e., we have $P\in \mathcal{P}$.
(Note that $\Pclo$ and $\Pcc$ are contained in $\mathcal{P}$ by definition.)
Thus, it remains to prove \ref{item:redPathsDoNotvisitV} and \ref{item:colorChangeImpliesUsingA}.

First, \ref{item:redPathsDoNotvisitV} follows by definition of $\Pclo$ and $\Pcc$. If a path enters $v$ using a red arc, this arc lies either clockwise between $\bclo$ and $\aclo$ or counterclockwise between $\bcc$ and $\acc$, see \Cref{fig:v} and Equation~\eqref{eq:neighborhood_v}. 
This is because the red arcs of $\delta(v)$ are precisely those that are counterclockwise between $\aclo$ and $\acc$; moreover, all arcs counterclockwise between $\bclo$ and $\bcc$ are arcs of $\phi^{-1}(a)$, which leave $v$. By definition of $\Pclo$, there is no path entering $v$ clockwise between $\bclo$ and $\aclo$ and using an arc of $\phi^{-1}(a)$ afterward. 
Analogously, by definition of $\Pcc$, there is no path entering counterclockwise between $\bcc$ and $\acc$ and then using an arc of $\phi^{-1}(a)$. This shows \ref{item:redPathsDoNotvisitV}.

Finally, \ref{item:colorChangeImpliesUsingA} follows because paths are non-crossing.
By \Cref{lem:only_one_change_of_color} a path can only change its color in $v$. Now, assume for the sake of deriving a contradiction that there is a path $P$ that enters $v$ via a blue arc and leaves $v$ via a red one not in $\phi^{-1}(a)$. Such a path either crosses $\Pclo$ (when it leaves clockwise between $\bclo$ and $\aclo$) or $\Pcc$ (when it leaves counterclockwise between $\bcc$ and $\acc$)---see also \Cref{fig:v} and Equation~\eqref{eq:neighborhood_v}---which is a contradiction.

\end{proof}

\subsection{Nice $s$-path partitionings are non-interleaving}\label{sec:non-interleaving}

We now show \Cref{lem:non-interleaving}, which we restate here for convenience.

\nonInterleaving*

\begin{proof}
Assume for the sake of deriving a contradiction that there are two interleaving sets $S^{t_1}$ and $S^{t_2}$, for some $t_1, t_2 \in T$, with $a_i, b_i \in S^{t_i}$ for $i \in [2]$ such that $a_1<a_2<b_1<b_2$.

Consider the paths $P_{a_1}$ and $P_{b_1}$, which both go to the same terminal $t_1$.
We consider the tri-coloring of $H$ induced by $P_{a_1}$, $P_{b_1}$, and $t_1$, see \Cref{fig:coloring_example} with $v=t_1$. 
By \Cref{lem:only_one_change_of_color}, a path can change the color only once and only at $t_1$. Because $t_1$ does not have any outgoing arcs, we get that no path can change its color. 
As paths are source-numbered, $P_{a_2}$ and $P_{b_2}$ start with arcs of different color.
Thus, their joint terminal $t_2$ needs to lie on either $P_{a_1}$ or $P_{b_1}$, because for all vertices $w$ not lying on one of these two paths, the arcs in $\delta(w)$ have all the same color.
However, as $t_2$ does not have any outgoing arcs and is disjoint from $t_1$, it cannot lie on either $P_{a_1}$ or $P_{b_1}$, leading to a contradiction.
\end{proof}

\section{Interval-discrepancy on non-interleaving partitions}\label{sec:discrepancy}

In this section we prove our two main discrepancy statements, \Cref{thm:discrepancy_main} and \Cref{thm:discrepancy_main_costs}.

\subsection{The case without costs (Proof of \Cref{thm:discrepancy_main})}

For convenience, we first recall \Cref{thm:discrepancy_main}.
\discrepancyMain*

In order to prove \Cref{thm:discrepancy_main}, we show the following bound on the discrepancy of the prefix intervals $[i]$ with $i\in [\ell]$.

\begin{theorem}\label{thm:prefix_discrepancy_main}
Consider a WPCS instance $(\ell,\mathcal{S}, d,y)$. 
 If $\mathcal{S}$ is non-interleaving, then there exists an integral selection  $z$ with
    \[
    D_{y,z}([i]) \leq  \frac{d_{\max}}{2}
    \]
 for all $i\in [\ell]$, where $d_{\max} \coloneqq \max_{S \in \mathcal{S}} d_S$.
Moreover, it can be computed in time $O(\ell)$. 
\end{theorem}
    
Before proving \Cref{thm:prefix_discrepancy_main}, we first observe that it implies \Cref{thm:discrepancy_main}.
For any interval $I=\{i, i+1,\dots,j\}\subseteq [\ell]$, we have 
\begin{align*}
D_{y,z} (I) = D_{y,z}([j]\setminus [i-1])\le \ D_{y,z} ([i-1]) +  D_{y,z} ([j]) \le d_{\max}.
\end{align*}
If $I$ is a circular interval, then either $I$ is an interval or the complement $[\ell]\setminus I$ of $I$ is an interval.
In the latter case we have $D_{y,z} (I) = D_{y,z}([\ell]\setminus I)  \le d_{\max}$.

\smallskip

Hence, it remains to prove \Cref{thm:prefix_discrepancy_main}, which we do in the remainder of this section.
We fix a WPCS instance $(\ell,\mathcal{S}, d,y)$ such that $\mathcal{S}$ is non-interleaving.

For a given integral selection  $z \in \{0,1\}^{\mathcal{S}}$, we define the values
\begin{equation*}
D^k \coloneqq y^d([k]) - z^d([k]) \quad \text{for } k\in \{0,\dots, \ell\}.
\end{equation*}
Then $D_{y,z}([k])=\abs{D^k}$.
We now show that the difference of $D^i$ and $D^j$ for $i,j \in S \in \mathcal{S}$ depends only on the restriction of $z$ to $S$ and not on any other entry of the vector $z$.
This is a  key consequence of the fact that $\mathcal{S}$ is non-interleaving.

\begin{lemma}\label{lem:main_consequence_non_interleaving}
Let $z \in \{0,1\}^{\mathcal{S}}$ be an integral selection.
Let $S\in \mathcal{S}$ and $i,j\in S$ with $i \le j$.
Then
\[
  D^j - D^{i-1} = \sum_{\substack{m\in S:\\ i \leq m \leq j}} (y^d_m -z^d_m).
\]
\end{lemma}
\begin{proof}
Consider a set $R\in \mathcal{S}\setminus \{ S \}$.
Because $\mathcal{S}$ is non-interleaving, we have either $R\cap \{ i,i+1,\dots, j\} = \emptyset$ or $R\subseteq  \{ i,i+1,\dots, j\}$.
In both cases we have $\sum_{m\in R: i\leq m\leq j} (y^d_m -z^d_m) =0$, and thus
\[
D^j - D^{i-1}  = \sum_{m=i}^j (y^d_m -z^d_m) = \sum_{U\in \mathcal{S}} \sum_{\substack{m\in U:\\ i\leq m\leq j}} (y^d_m -z^d_m) = \sum_{\substack{m\in S: \\ i\leq m\leq j}} (y^d_m -z^d_m).
\]
\end{proof}

We now describe the algorithm which we use to compute an integral selection $z$ with the desired properties.
The algorithm fixes the entries $z_j$ from $j=1$ to $j=\ell$ one by one in this order, while maintaining the desired prefix discrepancy bound $D^j\in [-\frac{d_{\max}}{2},\frac{d_{\max}}{2}]$.

Suppose the currently considered entry is $z_j$ with $j\in S\in \mathcal{S}$.
If, for some $k\in S$ with $k<j$, the entry $z_k$ has already been fixed to $1$ in an earlier iteration, the algorithm has to choose $z_j = 0$ in order to obtain an integral selection.
Similarly, if $j$ is the largest element of $S$, and we have chosen $z_k=0$ for all other elements $k$ of $S$, the algorithm must set $z_j=1$ to obtain an integral selection.
We will use \Cref{lem:main_consequence_non_interleaving} to prove that in both of these cases we have $\abs{D^j} \le \frac{d_{\max}}{2}$ (using that prior prefix discrepancies are fine, i.e., $|D^i|\leq \frac{d_{\max}}{2}$ for $i<j$). 
If neither of the above two cases applies, our algorithm greedily fixes the currently considered entry $z_j$ to ensure  $\abs{D^j} \le \frac{d_{\max}}{2}$. 
This is possible because $D^{j-1} \in [-\frac{d_{\max}}{2}, \frac{d_{\max}}{2}]$ and the two different values that $D^j$ can attain, depending on the choice of $z^j\in \{0,1\}$, differ by at most $d_S \le d_{\max}$ from each other, where one of the two values lies below $D^{j-1}$ and the other one lies above $D^{j-1}$.
A formal description of the algorithm is given by \Cref{alg:discrepancy}.
Note that the runtime of the algorithm is $O(\ell)$, since every iteration of the for-loop can be executed in constant time ($D^{j-1}$ does not need to be recomputed from scratch in iteration $j$ but can be obtained from the previously computed $D^{j-2}$ in constant time).

\begin{algorithm2e}[ht]
\DontPrintSemicolon
\caption{Computing an integral selection $z$ with small interval-discrepancy.}
\label{alg:discrepancy}
\For{$j=1, \dots, \ell$}{
Let $S\in \mathcal{S}$ such that $j \in S$.\\
If $z_k =1$ for some $k\in S$ with $k < j$, set $z_j \coloneqq 0$. \\
If $j$ is the largest element of $S$ and $z_k=0$ for all $k\in S$ with $k<j$, set $z_j \coloneqq 1$. \\
If none of the above two cases apply, set
\[
z_j \coloneqq 
\begin{cases}
1 & \text{ if } D^{j-1} + y^d_j - d_S \ge - \frac{d_{\max}}{2}\\
0 & \text{ otherwise,}
\end{cases}
\]
where $D^{j-1} \coloneqq y^d([j-1]) - z^d([j-1])$.
}
\Return $z$.
\end{algorithm2e}

The vector $z$ returned by \Cref{alg:discrepancy} is an integral selection by construction. 
It remains to prove that it fulfills the discrepancy bound claimed by \Cref{thm:prefix_discrepancy_main}.

\begin{lemma}
Let $z$ be the integral selection returned by \Cref{alg:discrepancy}.
For all $j\in \{0,\dots, \ell\}$ we have $D_{y,z}([j])= \abs{D^j} \le \frac{d_{\max}}{2}$.
\end{lemma}
\begin{proof}
We prove the lemma by induction on $j$, starting with $j=0$.
We have $D^0=0$ and hence assume $j>0$.
As in \Cref{alg:discrepancy}, we let $S\in \mathcal{S}$ such that $j\in S$.
We distinguish the same three cases as the algorithm.

\medskip

First, suppose $z_k =1$ for some $k\in S$ with $k < j$.
Let $i$ be the minimal element of $S$.
By \Cref{lem:main_consequence_non_interleaving}, we have
\begin{align*}
D^k=&\ D^{i-1} + y^d(\{m\in S: i \le m \le k\}) - d_S, \text{ and} \\
D^j=&\  D^{i-1} + y^d(\{m\in S: i \le m \le j\}) - d_S.
\end{align*}

Since $y$ is a selection, we have $y^d(\{m\in S: i <m \le j\})\leq d_S$.
This implies $D^j \leq D^{i-1}$.
Moreover, because $y^d \ge 0$ and $k < j$, we get $D^k \le D^j$.
But since $\abs{D^{i-1}} \leq \frac{d_{\max}}{2}$ and $\abs{D^k} \leq \frac{d_{\max}}{2}$ by induction, we get $\abs{D^j} \leq \frac{d_{\max}}{2}$.

\medskip

Now suppose  $j$ is the largest element of $S$ and $z_k=0$ for all $k\in S$ with $k<j$. 
Let $i$ be the minimal element of $S$.
Then \Cref{lem:main_consequence_non_interleaving} implies
\begin{align*}
D^j - D^{i-1} =& \sum_{m\in S\colon i\leq m \leq j} (y_m^d - z_m^d) = \sum_{m\in S} (y_m^d - z_m^d) = 0,
\end{align*}
where the second equation follows from the fact that $i$ is the smallest element of $S$ and $j$ the largest one, and the last equality holds due to $y(S)=1$ and $z(S)=1$ because $y$ and $z$ is a fractional and integral selection, respectively.
By the induction hypothesis, we conclude $\abs{D^j} = \abs{D^{i-1}} \le \frac{d_{\max}}{2}$.

\medskip

Finally, we assume that none of the above cases applies.
If we set $z_j=1$, then we have 
$D^j = D^{j-1} + y^d_j - z^d_j = D^{j-1} + y^d_j - d_S \ge -\frac{d_{\max}}{2}$, where we used the choice of $z_j$ in  \Cref{alg:discrepancy}.
Moreover, because $y^d_j = y_j d_S \le d_S$, we have $D^{j-1} \ge D^j $, which implies $D^j \le \frac{d_{\max}}{2}$ by the induction hypothesis.
If we set $z_j=0$, we have $D^j \ge D^{j-1}$, which implies $D^j \ge - \frac{d_{\max}}{2}$ by the induction hypothesis.
Moreover, by the choice of $z_j$ in \Cref{alg:discrepancy} we have $D^{j-1} + y^d_j - d_S \le  -\frac{d_{\max}}{2}$, implying $D^j = D^{j-1} + y^d_j \le d_S  -\frac{d_{\max}}{2} \le  \frac{d_{\max}}{2}$, where we used $d_S \le d_{\max}$.
\end{proof}

\subsection{The case with costs (Proof of \Cref{thm:discrepancy_main_costs})}\label{sec:proof_discr_cost}

We first restate \Cref{thm:discrepancy_main_costs} for convenience.
\discrepancyMainCosts*

As in the proof of \Cref{thm:discrepancy_main}, we start by bounding the discrepancy of the prefix intervals~$[i]$ with $i\in [\ell]$.

\begin{theorem}\label{thm:prefix_discrepancy_main_costs}
Consider a WPCS instance $(\ell,\mathcal{S}, d,y)$  and let $c \in \mathbb{Q}^{\ell}$.
 If $\mathcal{S}$ is non-interleaving, there exists an integral selection  $z$ with
     \begin{align*}
    D_{y,z}([i]) &\leq d_{\max} \quad \text{ for all }i\in [\ell]\text{ and}\\[2mm]
    c^Tz &\leq c^Ty,
    \end{align*}
where $d_{\max} \coloneqq \max_{S \in \mathcal{S}} d_S$.
Moreover, such a selection $z$ can be computed in time $O(\inp \! \cdot \ell)$. 
\end{theorem}

Before proving \Cref{thm:prefix_discrepancy_main_costs}, we first observe that it implies \Cref{thm:discrepancy_main_costs}.
For any interval $I=\{i, i+1,\dots,j\}\subseteq [\ell]$, we have 
\begin{align*}
D_{y,z} (I) = D_{y,z}([j]\setminus [i-1])\le \ D_{y,z} ([i-1]) +  D_{y,z} ([j]) \le 2 d_{\max}.
\end{align*}
If $I$ is a circular interval, then either $I$ is an interval or the complement $[\ell]\setminus I$ of $I$ is an interval.
In the latter case we have $D_{y,z} (I) = D_{y,z}([\ell]\setminus I)  \le 2 d_{\max}$.

To prove \Cref{thm:prefix_discrepancy_main_costs}, we first consider a special case of it where the given fractional selection $y$ is half-integral.
In this special case we will be able to achieve a stronger discrepancy bound.
Then we will show that the general case can be reduced to the special case at the cost of increasing the interval discrepancy by a factor of two.
The approach we employ to reduce to this special case is a common argument in discrepancy theory (see, e.g, \cite{lovasz_1986_discrepancy}, and \cite{bansal2022flow} for a recent application in a scheduling context).

\begin{lemma}\label{lem:half-integral}
 Consider a WPCS instance $(\ell,\mathcal{S}, d,y)$ and let $c \in \mathbb{Q}^{\ell}$.
 Suppose $|S|=2$ for all $S\in \mathcal{S}$ and $y_i = \frac{1}{2}$ for all $i\in [\ell]$.
 If $\mathcal{S}$ is non-interleaving, there exists an integral selection  $z$ with
     \begin{align*}
    D_{y,z}([i]) &\leq \frac{d_{\max}}{2} \quad \text{ for all }i\in [\ell]\text{ and}\\[2mm]
    c^Tz &\leq c^Ty,
    \end{align*}
where $d_{\max} \coloneqq \max_{S \in \mathcal{S}} d_S$.
      Moreover, such a selection $z$ can be computed in time $O(\ell)$.
\end{lemma}
\begin{proof}
We apply \Cref{thm:prefix_discrepancy_main} to obtain an integral selection $z \in \{0,1\}^{\ell}$ with $D_{y,z}([i]) \leq \frac{d_{\max}}{2}$ for all $i\in [\ell]$.
We show that either $z$ or the vector $\overline{z} \in \{0,1\}^{\ell}$, defined by
\[
\overline{z}_i \coloneqq 
\begin{cases}
1 & \text{ if } z_i =0 \\
0  &\text{ if } z_i =1,
\end{cases}
\]
has the desired properties.
First observe that for any integral selection $q\in \{0,1\}^{\ell}$ and any $j\in S \in \mathcal{S}$, we have
\[
y^d_j - q^d_j =
\begin{cases}
\frac{d_S}{2} & \text{ if } q_j =0 \\[2mm]
- \frac{d_S}{2} & \text{ if } q_j =1,
\end{cases}
\]
where we used $y_i =\frac{1}{2}$ for all $i \in [\ell]$.
This implies $y^d_j - \overline{z}^d_j = - (y^d_j - z^d_j)$ for  all $j\in [\ell]$.
Thus, we have for all $i\in [\ell]$:
\[
D_{y,\overline{z}}([i]) = \abs{\sum_{j\in [i]} (y^d_j - \overline{z}^d_j)} =  \abs{ - \sum_{j\in [i]} (y^d_j - z^d_j)} = D_{y,z}([i])  \le  \frac{d_{\max}}{2}.
\]
Moreover, we have $z + \overline{z} = 2 y$ and thus $\min\{ c^Tz, c^T\overline{z}\} \le c^Ty$.
We conclude that one of the integral selections $z$ or $\overline{z}$ has the desired properties. 
The runtime is $O(\ell)$, since \Cref{thm:prefix_discrepancy_main} gives an integral selection $z$ in that time and computing $\overline z$ and evaluating the cost is both also linear in $\ell$.   
\end{proof}

For $\alpha \in \mathbb{Q}$, we call a vector $y\in \mathbb{Q}^{\ell}$ an \emph{$\alpha$-integral vector} if every entry of $y$ is an integer multiple of $\alpha$.
As a consequence of \Cref{lem:half-integral} we obtain the following.

\begin{lemma}\label{lem:round_one_bit}
 Consider a WPCS instance $(\ell,\mathcal{S}, d,y)$ with $\mathcal{S}$ being non-interleaving, let $c \in \mathbb{Q}^{\ell}$, and let $k\in \mathbb{Z}_{> 0}$.
Given a $2^{-k}$-integral fractional selection $y$, we can compute in polynomial time a $2^{-k+1}$-integral fractional selection $\overline{y}$ such that
     \begin{align*}
    \abs{y^d([i]) - \overline{y}^d([i])} &\leq 2^{-k} \cdot d_{\max} \quad \text{ for all }i\in [\ell]\text{ and}\\[2mm]
    c^T\overline{y} &\leq c^Ty,
    \end{align*}
where $d_{\max} \coloneqq \max_{S \in \mathcal{S}} d_S$.
      Moreover, such a selection $\overline{y}$ can be computed in time $O(\ell)$. 
\end{lemma}
\begin{proof}
We construct a WPCS instance $(\hat{\ell}, \hat{\mathcal{S}}, \hat{d}, \hat{y})$ to which we will then apply \Cref{lem:half-integral}.
Let
\begin{equation*}
F \coloneqq \{ i\in [\ell] : y_i \text{ is not $2^{-k+1}$-integral}\},
\end{equation*}
and we denote by $\hat{\ell}\coloneqq |F|$ the number of elements in $F$.
Moreover, we define a renumbering $f \colon  F \to [\hat{\ell}]$ such that $f$ is the (unique) bijection preserving the order of elements, i.e., $f(i) < f(j)$ if and only if $i < j$.
Because $y$ is $2^{-k}$-integral and $y(S)=1$  for $S\in\mathcal{S}$, we have that $|F\cap S|$ is even for every set $S\in \mathcal{S}$. 
Hence, there is a partition $\mathcal{R}_{S}$ of $F\cap S$ into sets of size two, which we choose such that no two sets in $\mathcal{R}_S$ are interleaving.
This can be achieved, e.g., by considering the elements of $F\cap S$ in increasing order and repeatedly putting two consecutive elements in the same set in $\mathcal{R}_S$.
We define
\begin{equation*}
\hat{\mathcal{S}} \coloneqq \{ \{f(i), f(j)\} : \{i,j\} \in \mathcal{R}_S,  S\in \mathcal{S} \}.
\end{equation*}
Then $\hat{\mathcal{S}}$ is a partition of $[\hat{\ell}]$.
Moreover, $\hat{\mathcal{S}}$ is non-interleaving, because any two distinct sets from the same set $\mathcal{R}_S$ are non-interleaving by construction, and for $R_1 \in \mathcal{R}_{S_1}$ and $R_2 \in \mathcal{R}_{S_2}$ for distinct sets $S_1,S_2\in \mathcal{S}$, the sets $R_1 \subseteq S_1$ and $R_2\subseteq S_2$ are non-interleaving because $S_1$ and $S_2$ are non-interleaving.
For a set $\hat{S}=\{f(i),f(j)\}$ with $\{i,j\}\subseteq S\in \mathcal{S}$, we define $\hat{d}_{\hat S} \coloneqq d_S$. 
Moreover, we set $\hat{y}_i \coloneqq \frac{1}{2}$ for all $i\in [\hat{\ell}]$ and $\hat{c}_{f(i)} \coloneqq c_i$ for all $i\in [\hat{\ell}]$.

We apply \Cref{lem:half-integral} to obtain an integral selection $\hat{z}$ for the  WPCS instance $(\hat{\ell}, \hat{\mathcal{S}}, \hat{d}, \hat{y})$ with cost $\hat{c}$.
Hence, $D_{\hat{y},\hat{z}}([\hat{i}]) \leq \frac{d_{\max}}{2}$ for $\hat{i}\in [\hat{\ell}]$, and $\hat{c}^T\hat{z} \leq \hat{c}^T\hat{y}$.
Then we define
\begin{equation*}
\overline{y}_i \coloneqq
\begin{cases}
y_i & \text{ if } i \in [\ell]\setminus F \\
y_i + 2^{-k+1} \cdot \left(\hat{z}_{f(i)} -\hat{y}_{f(i)}\right) & \text{ if }i \in F.
\end{cases}
\end{equation*}

Because $\hat{z}$ is integral, the definition of $F$ implies that $\overline{y}$ is $2^{-k+1}$-integral.
Moreover, $\hat{c}^T\hat{z} \le \hat{c}^T \hat{y}$ implies $c^T\overline{y} \le c^Ty$.
Finally, we consider an index $i\in [\ell]$.
If $F\cap [i] = \emptyset$, we have $\overline{y}^d([i]) = y^d([i])$.
Otherwise, let $i_F \coloneqq \max \{f(j) : j\in F\cap [i]\}$.
Then 
\begin{equation*}
\abs{y^d([i]) - \overline{y}^d([i])}= 2^{-k+1 } \abs{\hat{y}^d([i_F]) - \hat{z}^d([i_F])} =  2^{-k+1 } \cdot  D_{\hat{y},\hat{z}}([i_F])  \le  2^{-k} d_{\max}.
\end{equation*}
By \Cref{lem:half-integral}, we obtain $\hat z$ in $O(\hat\ell)$, where $\hat\ell \leq \ell$ by definition. The subsequent definition of $\bar y$ can clearly be done in $O(\ell)$.
\end{proof}

By applying \Cref{lem:round_one_bit} repeatedly, we can round a $2^{-k}$-integral fractional selection to an integral one.
However, we might be given a fractional selection $y$ that is not $2^{-k}$-integral for any $k\in \mathbb{Z}_{>0}$.
To handle this case, we will use the following simple observation.

\begin{lemma}\label{lem:rounding_epsilon}
Consider a WPCS instance $(\ell,\mathcal{S}, d,y)$ with $\mathcal{S}$ being non-interleaving, and let $c \in \mathbb{Q}^{\ell}$.
Let $\epsilon > 0$ and $k_{\epsilon}\coloneqq \lceil \log_2 (\sfrac{\ell}{\epsilon})\rceil$.
Then, there exists a $2^{-k_{\epsilon}}$-integral fractional selection $\tilde{y}$ such that
     \begin{align*}
    \abs{y^d(U) - \tilde{y}^d(U)} &\leq \epsilon \cdot d_{\max} \quad \text{ for all }U \subseteq [\ell]\text{, and}\\[2mm]
    c^T\tilde{y} &\leq c^Ty,
    \end{align*}
Moreover, such a selection $\tilde{y}$ can be computed in time $O(\ell)$. 
\end{lemma}
\begin{proof}
For each $S\in \mathcal{S}$, let $i_S \in S$ be such that its cost is minimum, i.e., such that  $c_{i_S} \le c_j$ for all $j\in S$.
Then, for all $j\in S\setminus \{i_S\}$, let $\tilde y_j$ be the number obtained from rounding down $y_j$ to the next integer multiple of $2^{-k_{\epsilon}}$.
We define $\tilde y_{i_S} \coloneqq 1 - \sum_{j\in S\setminus \{i_S\}} \tilde y_j$.
Then $\tilde y$ is a fractional selection and, by the choice of the element $i_S \in S$, we have $c^T \tilde y \le c^Ty $.
Moreover, for every set $U\subseteq [\ell]$, we have $\abs{y(U) - \tilde y(U) } \le \ell \cdot 2^{-k_{\epsilon}} \le \epsilon$ and thus $\abs{y^d(U) - \tilde y^d(U) }\le \epsilon \cdot d_{\max}$.
\end{proof}

We now complete the proof of \Cref{thm:prefix_discrepancy_main_costs}.
By applying once \Cref{lem:rounding_epsilon} for some $\epsilon >0$ and applying \Cref{lem:round_one_bit} repeatedly $k_{\epsilon}= \lceil \log_2 (\sfrac{\ell}{\epsilon})\rceil$ times, we obtain an integral selection $z$ such that, for all $i\in [\ell]$,
\begin{equation}\label{eq:epsilon_discr_bound}
 \abs{y^d([i]) - z^d([i])} \leq \epsilon \cdot d_{\max} + \sum_{i=1}^{k_{\epsilon}} 2^{-i} \cdot d_{\max} 
 < (1+\epsilon) \cdot d_{\max}.
\end{equation}
If we choose $\epsilon > 0$ small enough, i.e., such that 
\begin{equation}\label{eq:epsilon_condition_existence}
\epsilon \cdot d_{\max} <  \abs{y^d([i]) -  \tilde z^d([i])} - d_{\max}
\end{equation}
for all integral selections $\tilde z$ and all $i\in [\ell] $ for which $\abs{y^d([i]) - \tilde z^d([i])} > d_{\max}$,
then \eqref{eq:epsilon_discr_bound} implies $\abs{y^d([i]) - z^d([i])} \leq d_{\max}$.
Because the right-hand side of \eqref{eq:epsilon_condition_existence} can attain only finitely many values, such an $\epsilon > 0$ does indeed exist.

To obtain a polynomial time algorithm, we exploit rationality of the entries of the vectors $y$ and $d$, and write $d_i = \frac{K_i}{L_i}$ and $y_i = \frac{M_i}{N_i}$ for integers $K_i, L_i, M_i, N_i$ for all $i\in [\ell]$.
Then we set
\begin{equation*}
\epsilon \coloneqq \frac{1}{d_{\max}}\cdot \prod_{i\in [\ell]} \frac{1}{L_i \cdot N_i}.
\end{equation*}
Note that for this choice of $\epsilon$, the number $k_{\epsilon} = \lceil \log_2 (\sfrac{\ell}{\epsilon})\rceil$ of applications of \Cref{lem:round_one_bit} is linearly bounded in the input size, which implies an overall runtime of $O(\inp \cdot \ell)$.
Moreover, \eqref{eq:epsilon_discr_bound} implies that, for every $i\in [\ell]$, we have
\begin{equation*}
\abs{y^d([i]) - z^d([i])} - d_{\max} < \epsilon \cdot d_{\max} = \prod_{i\in [\ell]} \frac{1}{L_i \cdot N_i}.
\end{equation*}
Because both $\abs{y^d([i]) - z^d([i])}$ and $d_{\max}$ are integer multiples of $\prod_{i\in [\ell]} \frac{1}{L_i \cdot N_i}$, this implies that we have $\abs{y^d([i]) - z^d([i])} \le d_{\max}$, concluding the proof of \Cref{thm:prefix_discrepancy_main_costs}.

\section{Concluding remarks}\label{sec:conclusions}

Recall that we required all numbers in our SSUF instance to be rational, so that we have finite input length, which puts us in the traditional computational model to talk about efficient algorithms.
We remark that even when allowing arbitrary real numbers (for the given flow $x$, the demands $d$, and the costs $c$), the existence of an unsplittable flow as in \Cref{thm:main} and \Cref{thm:main_cost} follows from our results.
This can be derived from \Cref{thm:main} and \Cref{thm:main_cost} by a standard continuity argument.

Moreover, we highlight that the capacity violation $d_{\max}$ of \Cref{thm:main}, which shows \Cref{conj:morellAndSkutella_weaker} in the special case of planar graphs, is tight.
This is the case even if there is only a single demand, as observed already in earlier work \cite{dinitz_1999_singlesource, morell_2022_single}.
To see this, consider an instance with a single terminal $t$ with demand $k$ and a flow $x$ that can be decomposed into $k$ arc-disjoint paths with flow value $1$ each.
Then $x(a)\le 1$ for all $a\in A$, but any unsplittable flow $\mathcal{P}$ will have value $k=d_ {\max}$ on some $s$-$t$ path.
\cite{morell_2022_single} gives a similar example showing that the lower bound on the unsplittable flow in \Cref{conj:morellAndSkutella_weaker} is tight, even in planar graphs.

In order to prove \Cref{conj:morellAndSkutella} for planar instances, one would need to strengthen \Cref{thm:main_cost} by improving the bound on the deviation $\abs{x(a)-\flow_{\mathcal{P}}(a)}$ of the unsplittable flow $\mathcal{P}$ and the given flow $x$ from $2d_{\max}$ to $d_{\max}$.
We expand on two natural strategies toward this and provide examples showing why these strategies do not work.

One natural strategy would be to strengthen the interval-discrepancy bound of \Cref{thm:discrepancy_main_costs} from $2d_{\max}$ to $d_{\max}$.
In the proof of \Cref{thm:discrepancy_main_costs} we showed that we can achieve a discrepancy bound of $d_{\max}$ for prefix intervals $[i]$ with $i\in [\ell]$ (see \Cref{thm:prefix_discrepancy_main_costs}).
One might hope to strengthen this to a bound of $\frac{d_{\max}}{2}$, as we did in the setting without costs (see \Cref{thm:prefix_discrepancy_main}).
However, in the setting with costs such a strengthening is impossible, as the following example shows.
Let $\epsilon \in (0,1]$ and consider the WPCS instance with $\ell = 2$, $\mathcal{S}=\{\{1,2\}\}$, demand $d_{\{1,2\}} = 1$, and  fractional selection $y$ with $y_1=\epsilon$ and $y_2 = 1-\epsilon$. 
If the cost vector $c$ fulfills $c_1=0$ and $c_2=1$, then the only integral selection $z$ with $c^Tz \le c^Ty= 1-\epsilon$ is the selection $z$ with $z_1=1$ and $z_2=0$.
Then $D_{y,z}([1])= 1-\epsilon$.
Because $d_{\max}=1$, and $\epsilon$ can be chosen arbitrarily close to zero, this shows tightness of \Cref{thm:prefix_discrepancy_main_costs}.

Another possible strategy to prove \Cref{conj:morellAndSkutella} for planar SSUF by strengthening \Cref{thm:discrepancy_main_costs} would be to use the reduction to the half-integral setting described in \Cref{sec:proof_discr_cost}, which comes at the cost of loosing a factor of two in the discrepancy bound.
One might then hope to achieve a $(y,z)$-interval-discrepancy bound of $\frac{d_{\max}}{2}$ for the half-integral special case.
However, this is in general impossible as the following example shows.
Let $\ell \coloneqq 6$ and $\mathcal{S}\coloneqq\{\textcolor{red!80!black}{\{1,6\}}, \textcolor{blue!80!black}{\{2,5\}}, \textcolor{green!60!black}{\{3,4\}}\}$.
Note that $\mathcal{S}$ is non-interleaving. 
Moreover, we define $\textcolor{red!80!black}{d_{\{1,6\}}\coloneqq 2}$, $\textcolor{blue!80!black}{d_{\{2,5\}}\coloneqq 1}$, $\textcolor{green!60!black}{d_{\{3,4\}}\coloneqq2}$ and $y_i \coloneqq \frac{1}{2}$ for all $i\in [\ell]$.
Then
\[
\textcolor{red!80!black}{y^d_1 = 1} \quad \quad 
\textcolor{blue!80!black}{y^d_2 = 0.5} \quad \quad
\textcolor{green!60!black}{y^d_3 = 1} \quad \quad 
\textcolor{green!60!black}{y^d_4 = 1} \quad \quad 
\textcolor{blue!80!black}{y^d_5 = 0.5} \quad \quad 
\textcolor{red!80!black}{y^d_6 = 1} 
\]
and $d_{\max}=2$.
Now consider an integral selection $z$.
By symmetry, we may assume $z_1=1$ and $z_6=0$.
If $z_3=1$ and $z_4=0$, we have 
\[
\textcolor{red!80!black}{z^d_1 = 2} \quad \quad 
\textcolor{blue!80!black}{z^d_2 \ge 0} \quad \quad
\textcolor{green!60!black}{z^d_3 = 2} \quad \quad
\textcolor{green!60!black}{z^d_4 = 0} \quad \quad
\textcolor{blue!80!black}{z^d_5 \le 1} \quad \quad
\textcolor{red!80!black}{z^d_6 = 0}
\]
and thus $D_{y,z}(\{1,2,3\}) \ge  1.5 > \frac{d_{\max}}{2}$.
Otherwise, we have $z_3=0$ and $z_4=1$.
Then we either have
\[
\textcolor{red!80!black}{z^d_1 = 2} \quad \quad 
\textcolor{blue!80!black}{z^d_2 = 1} \quad \quad 
\textcolor{green!60!black}{z^d_3 = 0} \quad \quad 
\textcolor{green!60!black}{z^d_4 = 2} \quad \quad 
\textcolor{blue!80!black}{z^d_5 = 0} \quad  \quad 
\textcolor{red!80!black}{z^d_6 = 0}
\]
or 
\[
\textcolor{red!80!black}{z^d_1 = 2} \quad  \quad 
\textcolor{blue!80!black}{z^d_2 = 0} \quad \quad 
\textcolor{green!60!black}{z^d_3 = 0} \quad \quad 
\textcolor{green!60!black}{z^d_4 = 2} \quad \quad 
\textcolor{blue!80!black}{z^d_5 = 1} \quad \quad 
\textcolor{red!80!black}{z^d_6 = 0}. 
\]
In the first case, $D_{y,z}(\{1,2\}) = 1.5 > \frac{d_{\max}}{2}$ and in the second case $D_{y,z}(\{4,5\}) = 1.5 > \frac{d_{\max}}{2}$.
We conclude that in general it is indeed impossible to achieve a $(y,z)$-interval discrepancy of at most $\frac{d_{\max}}{2}$ on half-integral instances.

\appendix

\section{Hardness of SSUF with lower bounds on cyclic graphs}
\label{sec:hardness_cyclic_graphs}

We start by observing that in case of general (cyclic) graphs, \Cref{conj:morellAndSkutella} and even \Cref{conj:morellAndSkutella_weaker} cannot hold as stated, because the desired unsplittable flow might not exist. 
\begin{observation}
Given some constant $\lambda >0$, there is an SSUF instance $(G=(V,A),s,T,d,x)$ such that there exists no unsplittable flow $\mathcal{P}=\{P^t\}_{t\in T}$ with
\[ x(a) - \lambda d_{\max} \leq \flow_{\mathcal{P}}(a)  \quad\text{for all } a\in A.\]
\end{observation}

Consider the graph depicted in \Cref{fig:lower_bounds_and_cycles} with one commodity of demand $d$. The fractional flow $x$ on all black, thin arcs is $\frac{d}{2}$, while it is $d(\lambda+1)$ on all blue thick arcs. Note that this flow is a convex combination of two walks. Still, in this case there is no path, and actually not even a walk, that can send positive flow along both triangles, which means that, for any unsplittable flow, some lower bounds are violated heavily.

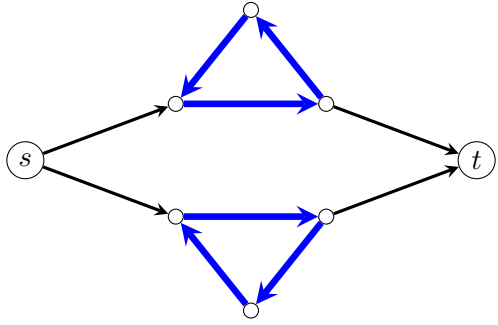
\begin{figure}
\begin{center}
\begin{tikzpicture}[rotate=-90,
terminals/.style={circle,draw=black, inner sep=0pt, minimum size=14pt,},
vertices/.style={circle, draw=black,inner sep=2pt},
arcs/.style={line width=1.2pt,-stealth},
cyclearcs/.style={line width =2.5pt,-stealth, blue}
]

\begin{scope}
\node[terminals] (s) at (0,0) {$s$};
\node[vertices] (v1) at (-0.75,2) { };
\node[vertices] (v2) at (-0.75,4) { };
\node[vertices] (v3) at (-2,3) { };
\node[vertices] (v4) at (0.75,2) { };
\node[vertices] (v5) at (0.75,4) { };
\node[vertices] (v6) at (2,3) { };
\node[terminals] (t) at (0,6) {$t$};

\begin{scope}[arcs]
\draw (s) -- (v1);
\draw[cyclearcs] (v1) -- (v2);
\draw[cyclearcs] (v2) -- (v3);
\draw[cyclearcs] (v3) -- (v1);
\draw (s) -- (v4);
\draw[cyclearcs] (v4) -- (v5);
\draw[cyclearcs] (v5)-- (v6);
\draw[cyclearcs] (v6) -- (v4);
\draw (v2) -- (t);
\draw (v5) -- (t);
\end{scope}
\end{scope}

\end{tikzpicture}

 \end{center}
\caption{\label{fig:lower_bounds_and_cycles}
The figure shows an example, where no unsplittable flow can satisfy any $O(d_{\max})$ violation of the lower bounds. Vertex $s$ is the source and $t$ is the terminal. The fractional flow splits half-half between the up and down option and then cycles often along the blue, thick triangles. The unsplittable flow can only choose one of the two options and thus violates the lower bounds heavily. 
}
\end{figure}

Thus, when considering general graphs, the nature of the problem substantially changes. Instead of proving existence of certain unsplittable flows under the assumption that a splittable flow exists, we consider a different variant of the problem where we are given lower bounds $(\ell(a))_{a \in A}$ instead.

In acyclic graphs we can use \Cref{thm:main} to either decide that there exists no unsplittable flow $\mathcal{P}=\{P^t\}_{t\in T}$ with $\ell(a) \leq \flow_{\mathcal{P}}(a)$ for all $a \in A$, or find an unsplittable flow $\mathcal{P}=\{P^t\}_{t\in T}$ with $\ell(a) - d_{\max} \leq \flow_{\mathcal{P}}(a)$.
To this end we simply check if there exists a splittable flow satisfying the lower bounds  $(\ell(a))_{a \in A}$.
If no such splittable flow exists, there is also no unsplittable flow satisfying the lower bounds.
Otherwise, we obtain a splittable flow $x$ satisfying the given lower bounds, and we can apply \Cref{thm:main} to $x$ to find an unsplittable flow that violates the lower bounds by at most $d_{\max}$.

While \Cref{thm:main} cannot be extended to cyclic graphs, as we observed above, one might still hope for an efficient algorithm that either decides that no unsplittable flow satisfying the lower bounds exists, or finds one that violates the lower bounds by at most $O(d_{\max})$.
We next show that this is impossible for general (not necessarily planar) instances unless $\P=\NP$.

\medskip

To this end, we reduce from the $2$-Vertex-Disjoint Paths problem ($2$-VDP), which is well-known to be NP-complete~\cite{fortune1980directed}.
In a $2$-VDP instance, we are given a directed graph $G=(V,A)$ and two pairs of vertices $(p_1,q_1)$ and $(p_2,q_2)$.
The task is to decide whether there are two vertex-disjoint paths $Q_1$ and $Q_2$ from $p_1$ to $q_1$ and from $p_2$ to $q_2$, respectively.
Given a $2$-VDP instance $I=(G,(p_1,q_1),(p_2,q_2))$, together with a desired number $k\in \mathbb{Z}_{\geq 1}$ of terminals, we create an auxiliary graph $\overline{G} = (\overline{V}, \overline{A})$ together with a source $s\in V$, $k$ terminals $T\subseteq \overline{V}$, and demands $d$ as follows:
\begin{itemize}
\item $\overline{G}$ is obtained from $G$ by adding an arc $(q_1,p_2)$, and by adding $k$ additional vertices $T=\{t_1,\dots, t_k\}$, which will be the terminals, together with arcs $(q_2,t_i)$ for all $i\in [k]$;
\item all demands are unit, i.e., $d_t=1$ for all $t\in T$;
\item $s=p_1$ is the source.
\end{itemize}
See \Cref{fig:hardness} for an illustration of this construction.
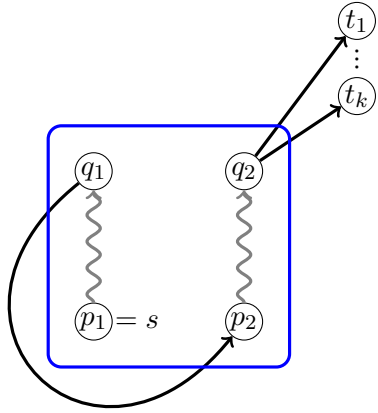
\begin{figure}
 \begin{center}
\usetikzlibrary{fit}
            \begin{tikzpicture}[
                terminals/.style={circle,draw=black, inner sep=0pt, minimum size=14pt,},
                vertices/.style={circle, draw=black,inner sep=2pt},
                box/.style = {draw,blue,line width=1.2pt,inner sep=10pt,rounded corners=5pt}]

                \node[terminals] (p1) at (2,2) {$p_1$}; 
                \node[terminals] (q1) at (2,4) {$q_1$}; 
                \node[terminals] (p2) at (4,2) {$p_2$}; 
                \node[terminals] (q2) at (4,4) {$q_2$}; 
                \node[terminals] (t1) at (5.5,6) {$t_1$};
                \node[rotate=90] at (5.5,5.5) {...};
                \node[terminals] (tk) at (5.5,5) {$t_k$};

                \node (s) at ($(p1) + (1.5em,0)$) {$=s$};

                \draw[gray,line width=1.2pt,decorate,decoration=snake, ->] (p1) to (q1);
                \draw[gray,line width=1.2pt,decorate,decoration=snake, ->] (p2) to (q2);
                 \draw (q1) edge[to path={.. controls (-0.5,2) and (2,-0.5).. (\tikztotarget)}, ->, line width=1.2pt]  (p2);
                \draw[->, line width=1.2pt] (q2) -- (t1);
                \draw[->, line width=1.2pt] (q2) -- (tk);
            
                \node[box,fit=(p1)(q2)(q1)(p2)] {};
            \end{tikzpicture} 

 \end{center}
\caption{\label{fig:hardness}
Given a $2$-VDP instance $I=(G,(p_1,q_1),(p_2,q_2))$, we obtain the corresponding graph $\overline{G}$ from $G$ by adding an arc $(q_1,p_2)$, adding terminals $t_1,\dots, t_k$, and adding for each terminal $t_i$ an arc $(q_2,t_i)$.
Moreover, the vertex $p_1$ is the source.
The idea is that there are two vertex-disjoint paths (as depicted in gray) if and only if there is a path from $s$ to a terminal using the arc $(q_1,p_2)$.     
}
\end{figure}
We denote by $\overline{I}=(\overline{G},s,T,d)$ the resulting single-source unsplittable flow instance without fractional flow vector $x$, as our statements here do not depend on a specific fractional flow or even its existence.

We then have the following immediate relation between the instance $I$ and $\overline{I}$, which exhibits a gap in terms of how much flow can be sent over the arc $(q_1,p_2)$ in $\overline{I}$, depending on whether $I$ is feasible.
\begin{proposition}\label{prop:twoVDPtoSSUF}
Let $I=(G,(p_1,q_1),(p_2,q_2))$ be a $2$-VDP instance, and $\overline{I}=(\overline{G},s,T,d)$ be the corresponding SSUF instance (without fractional flow $x$).
Then we have the following implications:
\begin{enumerate}[label=(\alph*)]
\item\label{item:caseIIsYes} $I$ is a yes instance $\implies$ $\exists$ unsplittable flow $\mathcal{P}=\{P^t\}_{t\in T}$ for $\overline{I}$ with $\flow_{\mathcal{P}}((q_1,p_2)) = k$.
\item\label{item:caseIIsNo} $I$ is a no instance $\implies$ $\nexists$ unsplittable flow $\mathcal{P}=\{P^t\}_{t\in T}$ for $\overline{I}$ with $\flow_{\mathcal{P}}((q_1,p_2)) > 0$.
\end{enumerate}
\end{proposition}
\begin{proof}
If $I$ is a yes instance, there is an $s$-$q_2$ path $P$ in $\overline{G}$ using the arc $(q_1,p_2)$.
By setting each $P^{t}$ for $t\in T$ to be the path obtained by appending the arc $(q_2,t)$ to $P$, we get the desired result.

We prove the contrapositive of the second statement.
Hence, if there is an unsplittable flow $\mathcal{P}$ with $\flow_{\mathcal{P}}((q_1,p_2)) > 0$, then at least one path $P\in \mathcal{P}$ must use the arc $(q_1,p_2)$.
The path $P$ contains a $p_1$-$q_1$ path and a $p_2$-$q_2$ path that are vertex-disjoint, which implies that $I$ is a yes instance.
\end{proof}

There are different ways to interpret \Cref{prop:twoVDPtoSSUF}.
One natural way is that, in (cyclic) single-source unsplittable flow problems, it is hard to distinguish between the case where there is an unsplittable flow fulfilling all lower bounds and the case where there is no unsplittable flow even if we allow for a heavy violation of the lower bounds. 
\begin{corollary}\label{cor:gaphardness}
Let $k\in \mathbb{Z}_{\ge 1}$ and let $\lambda \in \mathbb{R}_{\ge 0}$ such that $\lambda < k$ be constants.
Then it is NP-hard to distinguish between the following two cases for SSUF instances on general (cyclic) graphs with $k$ terminals:
\begin{enumerate}
    \item\label{item:gapHardCaseFulfillLowerBounds} There is an unsplittable flow satisfying all lower bounds, i.e., an unsplittable flow $\mathcal{P}=\{P^t\}_{t\in T}$ with $\ell(a) \leq \flow_{\mathcal{P}}(a)$ for all $a \in A$.\label{it:unsplittable_flow}
    \item\label{item:gapHardCaseViolateRelaxBounds} There is no unsplittable flow satisfying all lower bounds with a violation of up to $\lambda d_{\max}$, i.e., there is no unsplittable flow $\mathcal{P}=\{P^t\}_{t\in T}$ with $\ell(a) -\lambda d_{\max} \leq \flow_{\mathcal{P}}(a)$ for all $a \in A$.\label{it:no_unsplittable_flow_with_lower_bound_violation}
\end{enumerate}
\end{corollary}
\begin{proof}
The hardness follows by considering SSUF instances $\overline{I}=(\overline{G},s,T,d)$ (without fractional flow vector $x$) stemming from $2$-VDP instances $I=(G,(p_1,q_1),(p_2,q_2))$.
Moreover, we assign a lower bound of $\ell((q_1,p_2))=k$ to the arc $(q_1,p_2)$ and lower bounds of $0$ to all other arcs.
Then the two cases~\ref{item:gapHardCaseFulfillLowerBounds} and~\ref{item:gapHardCaseViolateRelaxBounds} of \Cref{cor:gaphardness} correspond to~\ref{item:caseIIsYes} and~\ref{item:caseIIsNo} of \Cref{prop:twoVDPtoSSUF}, respectively.
Hence, they allow for distinguishing yes from no instances of $2$-VDP, which is \NP-hard.
\end{proof}
Note that the reduction and statement only hold true for general (cyclic) graphs because $2$-VDP is polynomial-time solvable in acyclic graphs.
 
\printbibliography

\end{document}